\documentclass[12pt,a4paper]{article}
\usepackage[latin1]{inputenc}
\usepackage{amsmath}
\usepackage{amsfonts}
\usepackage{amssymb}
\usepackage{amsthm}
\usepackage[francais, english]{babel}
\usepackage{graphicx}
\usepackage{graphics}
\usepackage[T1]{fontenc}
\usepackage{float}
\usepackage{placeins}
\graphicspath{{images/}}
\usepackage{fancyhdr}
\usepackage{mathpazo}
\newcommand{\nm}{\noalign{\smallskip}}
\newcommand{\ds}{\displaystyle}

\fancypagestyle{plain}{ 
\fancyhead{} 

}

\newtheorem{prop}{Proposition}[section]
\newtheorem{cor}[prop]{Corollary}
\newtheorem{lem}[prop]{Lemma}
\newtheorem{thm}[prop]{Theorem}
\numberwithin{equation}{section} \numberwithin{figure}{section}

\title{Mathematical modeling of
fluorescence diffuse optical imaging of cell membrane potential
changes
\thanks{\footnotesize This work was supported  by ERC Advanced
Grant Project MULTIMOD--267184.}}

\author{Habib Ammari\thanks{\footnotesize Department of Mathematics and Applications,
Ecole Normale Sup\'erieure, 45 Rue d'Ulm, 75005 Paris, France
(habib.ammari@ens.fr, laure.giovangigli@ens.fr).} \and Josselin
Garnier\thanks{\footnotesize Laboratoire de Probabilit\'es et
Mod\`eles Al\'eatoires \& Laboratoire Jacques-Louis Lions,
Universit\'e Paris VII, 75205 Paris Cedex 13, France
(garnier@math.jussieu.fr).} \and Laure
Giovangigli\footnotemark[2]}

\date{}

\begin{document}

\maketitle


\begin{abstract}
The aim of this paper is to provide a mathematical model for
spatial distribution of membrane electrical potential changes by
fluorescence diffuse optical tomography. We derive the resolving
power of the imaging method in the presence of measurement noise.
The proposed mathematical model can be used for cell membrane
tracking with the resolution of the optical microscope.
\end{abstract}

\bigskip

\noindent {\footnotesize Mathematics Subject Classification
(MSC2000): 35R30, 35B30.}


\tableofcontents


\selectlanguage{english}

\section{Introduction}

The propagation of light through a highly scattering medium with
low absorption is well described by the diffusion equation
\cite{diffusion}. Diffuse optical imaging techniques measure the
spatially-dependent absorption and scattering properties of a
tissue. A light source illuminates the tissue, and detectors
measure the intensity of the exiting light at the boundary of the
tissue, after it underwent multiple scattering and absorption. One
can use these measurements to reconstruct, from the diffusion
equation, a map of the optical parameters of the studied
biological tissue \cite{planar, john}.



Diffuse optical imaging techniques use near infrared light,
because absorption by biological tissue is minimal at these
wavelengths, and one can then produce images deep in living
subjects or samples, up to several centimeters.

These techniques can be used to image fluorescing targets,
fluorophores, in tissues. When excited by light at a specific
wavelength, fluorophores emit light at a different wavelength in
order to decay to their ground state. Measurements of emitted
light exiting at the boundary of the tissue, combined with
measurements of residual excitation light from sources, after it
went through the tissue,  provide an insight of the tissue optical
properties. More precisely, these measurements allow to
reconstruct a map of the tissue optical parameters, the
distribution of fluorophore concentration, and fluorophore
lifetime, the time they spent in their excited state before
emitting light \cite{3, model2}. The fluorescent indicators, which
can be chosen with excitation and emission wavelengths in the near
infrared light spectrum, accumulate in specific areas. With such
techniques, one can then localize proteins, cells or diseased
tissues, visualize in vivo biological processes, and obtain
measurements of the concentration in tissues of important
physiological markers, such as oxygenated hemoglobin \cite{imag1,
vasilis, leary}. Detailed structural information as well as
indications of pathology can be obtained from these images.

In this paper, we mathematically formulate the imaging problem of
the spatial distributions of the transmembrane potential changes
induced in cells by applied external electric fields. The use of
optical detection methods for the measurement of fluorescence
response to membrane electric fields was reported in the early
1970s. Since then, considerable advances have been reported
\cite{7}. In \cite{9}, it has been demonstrated experimentally
that membrane potential changes can be imaged with the resolution
of the optical microscopy. The key feature of this system is the
combined use of an external electric field and fluorescence
tomography. The fluorescent indicators are designed in such a way
they respond linearly to the electrical potential jump across the
membrane. The application of the electric field enhances the
membrane fluorescence imaging.

The purpose of this paper is threefold. We first provide and
analyze a mathematical model for optical imaging of changes in
membrane electric potentials. Then we propose, in the linearized
case where the shape of the cell is a perturbation of a disk, an
efficient direct imaging technique based on an appropriate choice
of the applied currents. An iterative imaging algorithm for more
complex shapes is also suggested. Finally, we estimate the
resolving power of the proposed imaging algorithm in the presence
of measurement noise. In a forthcoming work, we will use the
proposed algorithm for implementing tracking approaches capable of
imaging the behavior of single or cluttered live cells.

Our main results in this paper can be summarized as follows. Let
$C$ be the cell and let $\Omega$ be the background domain. Given
an optical excitation $g$, the emitted light fluence is
$\Phi_{\mathrm{emt}}^{\,g}$, the solution to the diffusion
equation (\ref{summ2}) with $\Phi_{\mathrm{exc}}^{\,g}$ defined by
(\ref{summ1}) and $c_{\mathrm{flr}}$ being the concentration of
fluorophore supported on the cell membrane $\partial C$. Equation
(\ref{eq:cf}) gives the relation between the function
$c_{\mathrm{flr}}$ and the electric potential $u$ defined by
(\ref{eq:u2}). In order to image the cell membrane $\partial C$,
we establish identity (\ref{eq:pbinv}) and linearize in Theorem
\ref{prop439} relation (\ref{eq:cf}) for $\partial C$ being a
perturbation of a disk. Proposition \ref{proplsa} gives the least
squares estimate of the cell membrane perturbation. Introducing
the signal-to-noise ratio in (\ref{defsnr}), where $\sigma$ models
the measurement noise amplitude and $\epsilon$ corresponds to the
order of magnitude of the cell membrane perturbation,  we derive
in Theorem \ref{thmresolving} the resolving power of the imaging
method. Theorem \ref{proplsa2}, which is our main result in this
paper, provides expressions for the reconstructed modes in the
cell membrane perturbation in the presence of measurement noise
under physical assumptions on the size of the cell and the value
of the used frequency. A generalization of the linearization
procedure for arbitrary-shaped cell membranes is provided in
Proposition \ref{prop441} and the reconstruction of perturbations
of arbitrary-shaped cell membranes is  formulated as a
minimization problem, where the data is appropriately chosen in
order to maximize the resolution of the reconstructed images.

\section{Governing model}

We consider a  cell, that we want to image. We inject fluorescent
indicators, which stick only on the cell membrane \cite{10}. These
markers are chosen so that their concentration responds linearly
to the potential jump across the membrane, when the cell is
immersed in an external electric field \cite{9}. We apply such en
external electric field at the boundary of our domain and use
fluorescence optical diffuse tomography to reconstruct the
position and shape of the membrane.

\subsection{Coupled diffusion equations}

A sinusoidally modulated near infrared monochromatic light source
$g$, located at the boundary $\partial \Omega$ of the examined
domain $\Omega$, launches an excitation light fluence
$$\phi_{\mathrm{exc}}= \Phi_{\mathrm{exc}}(x,\omega) \, e^{i\omega
t},$$ at the wavelength $\lambda_{\mathrm{exc}}$, into $\Omega$.
At time $t$ and point $x$, $\phi_{\mathrm{exc}}$ represents the
average photon density, due to excitation by the source
oscillating at frequency $\omega$. After it undergoes multiple
scattering and absorption, this light wave reaches the fluorescent
markers, which are accumulated on $\partial C$, the membrane of
the cell $C$. The excited fluorophores emit a wave
$$\phi_{\mathrm{emt}}= \Phi_{\mathrm{emt}}(x,\omega)\, e^{i\omega
t},$$ at the wavelength $\lambda_{\mathrm{emt}}$. The intensity of
the emitted wave is proportional to the intensity of the
excitation wave, when it reaches the fluorescent molecule. The
emitted waves pass through the absorbing and scattering domain and
are detected at the boundary $\partial \Omega$.

In the near infrared spectral window, the propagation of light in
biological tissues can be modeled by the diffusion equation, which
is a limit of the radiative transport equation when the transport
mean free path is much smaller than the typical propagation
distance. Our model can therefore be described by the following
coupled diffusion equations completed by Robin boundary conditions
\cite{model1,model2,model3,model4}:
\begin{displaymath}
\left \{
\begin{array}{ll}
- \nabla \cdot \left ( D_{\mathrm{exc}}(x) \nabla
\Phi_{\mathrm{exc}}(x,\omega) \right ) + \left
(\mu_{\mathrm{exc}}(x) + \displaystyle \frac{i \omega}{c}\right )
\Phi_{\mathrm{exc}}(x,\omega) = 0 \quad& \textrm{in}\, \Omega, \\
\vspace{0.2cm} \displaystyle \ell_{\mathrm{exc}} \frac{\partial
\Phi_{\mathrm{exc}}}{\partial \nu}(x,\omega) +\displaystyle  \,
\Phi_{\mathrm{exc}}(x,\omega) = g(x) & \textrm{on} \,\partial
\Omega,
\end{array}
\right .
\end{displaymath}

\begin{displaymath}
\left \{
\begin{array}{l}
- \nabla \cdot \left ( D_{\mathrm{emt}}(x) \nabla
\Phi_{\mathrm{emt}}(x,\omega) \right ) + \left
(\mu_{\mathrm{emt}}(x) + \displaystyle \frac{i \omega}{c}\right )
\Phi_{\mathrm{emt}}(x,\omega) \\ \nm \hspace{3cm} \displaystyle =
\gamma(x,\omega)
\, \Phi_{\mathrm{exc}}(x,\omega) \quad \textrm{in}\, \Omega,  \\
\vspace{0.2cm} \displaystyle \ell_{\mathrm{emt}} \frac{\partial
\Phi_{\mathrm{emt}}}{\partial \nu}(x,\omega) +\displaystyle
 \, \Phi_{\mathrm{emt}}(x,\omega) = 0 \quad \textrm{on}
\,\partial \Omega .
\end{array}
\right .
\end{displaymath}

\noindent Here,

\begin{itemize}

\item $\nu$ denotes the outward normal at the boundary $\partial
\Omega$;

\item $c$ denotes the speed of light in the medium;

 \item $D_{\mathrm{exc}}$ and $\mu_{\mathrm{exc}}$
(respectively $D_{\mathrm{emt}}$ and $\mu_{\mathrm{emt}}$) denote
the photon diffusion and absorption coefficient at wavelength
$\lambda_{\mathrm{exc}}$ (respectively $\lambda_{\mathrm{emt}}$)
over the speed of light $c$. Assuming that the scattering is
isotropic, they can be expressed, for $i= \mathrm{exc},
\mathrm{emt}$, as follows:
\begin{displaymath}
D_{i}(x) = \displaystyle \frac {1} {d(\mu_{a,i}(x) +
\mu_{\mathrm{flr},i}(x) + \mu'_{s,i}(x))}\quad\textrm{and} \quad
\mu_{i}(x) = \mu_{a,i}(x) + \mu_{\mathrm{flr},i}(x)\,,
\end{displaymath}

where

\begin{itemize}

\item $\mu_{a,i}$ denotes the absorption coefficient, due to
natural chromophores of the medium, at wavelength $\lambda_{i}$;

\item $\mu_{\mathrm{flr},i}$ denotes the absorption coefficient,
due to fluorophores, at wavelength $\lambda_{i}$. This absorption
coefficient is proportional to the fluorophore concentration
$c_\mathrm{flr}(x)$. The proportionality coefficient,
$\varepsilon_{\mathrm{exc}}$, is the fluorophore extinction
coefficient at wavelength $\lambda_{i}$;

\item $\mu'_{s,i}$ denotes the reduced scattering coefficient at
wavelength $\lambda_{i}$; its inverse is the transport mean free
path.

\item $\ell_i$ is the extrapolation length. It is computed from
the radiative transport theory \cite{rossum} and is proportional
to the transport mean path. The multiplicative function depends on
the index mismatch between the scattering medium in $\Omega$ and
the surroundings.

\item $d$ is the space dimension;

\end{itemize}

\item $\gamma$ is given by
\begin{equation} \label{gamma}
\gamma(x,\omega) = \displaystyle \frac{\eta
\,\mu_{\mathrm{flr},\mathrm{exc}}(x)}{1-i \omega \tau(x)} =
\frac{\eta \varepsilon_{\mathrm{exc}}\,c_{\mathrm{flr}}(x)}{1-i
\omega \tau(x)},
\end{equation}
with $\eta$ and $\tau$ being respectively the fluorophore's
quantum efficiency and fluorescence lifetime.
\end{itemize}

\subsection{Model assumptions}

Let $\Omega$ be the background domain and let $C \Subset \Omega$
denote the cell. From now on, the space dimension $d$ is equal to
$2$ or $3$ and $\Omega$ and $C$ are bounded $\mathcal{C}^2$-
domains.

The fluorophores are only located on the cell membrane $\partial
C$, their concentration $c_\mathrm{flr}(x)$ is zero, except on
$\partial C$. We neglect their contribution to the absorption and
diffusion coefficient, that is,
\begin{displaymath}
D_{i}(x) = \displaystyle\frac{1}{d(\mu_{a,i}(x) +
\mu'_{s,i}(x))}\quad \textrm{and} \quad \mu_{i}(x) = \mu_{a,i}(x).
\end{displaymath}

In the near infrared spectral window, the absorption coefficient
is much smaller than the reduced scattering  coefficient. This is,
besides, one of the conditions to approximate the light
propagation in the medium  by the diffusion equation.

We can approximate the diffusion coefficients at the excitation
and emission wavelength as follows:
\begin{displaymath}
D_{i}(x) = \displaystyle\frac{1}{d \mu'_{s,i}(x)} .
\end{displaymath}

We consider that the optical parameters are constant in the domain
$\Omega$ and do not depend on the wavelength of the propagating
light. Hence, for $i=\mathrm{exc},\mathrm{emt}$,
\begin{displaymath}
D_{i}(x) = D_{i} = D = \displaystyle\frac{1}{d \mu'_{s}}, \quad
\mu_{i}(x) = \mu_{i} = \mu = \mu_{a}, \quad \mbox{and}\quad
\ell_i(x)=\ell_i = \ell .
\end{displaymath}

We consider that the fluorophore's fluorescence lifetime $\tau$ is
constant. From (\ref{gamma}) it follows that $\gamma$ depends on
the position $x$ only through $\mu_{\mathrm{flr}}(x)$, and more
specifically $c_{\mathrm{flr}}(x)$. It can then be written as
follows:
\begin{displaymath}
\gamma(x,\omega) = \tilde{\gamma}(\omega) \,c_{\mathrm{flr}}(x)
\quad \mbox{with} \quad \tilde{\gamma}(\omega) = \frac{\eta
\varepsilon_{\mathrm{exc}}}{1-i \omega \tau}.
\end{displaymath}

The coupled diffusion equations and their boundary conditions then
become
\begin{equation} \label{summ1}
\left \{
\begin{array}{ll}
- D \Delta \Phi_{\mathrm{exc}}^{\,g}(x,\omega) + \left (\mu +
\displaystyle \frac{i \omega}{c}\right )
\Phi_{\mathrm{exc}}^{\,g}(x,\omega) = 0 \quad& \textrm{in} \, \Omega , \\
\vspace{0.2cm} \displaystyle \ell \frac{\partial
\Phi_{\mathrm{exc}}^{\,g}}{\partial \nu}(x,\omega) +\displaystyle
\, \Phi_{\mathrm{exc}}^{\,g}(x,\omega) = g(x) & \textrm{on}
\,\partial \Omega ,
\end{array}
\right .
\end{equation}
\begin{equation} \label{summ2}
\left \{
\begin{array}{ll}
- D \Delta \Phi_{\mathrm{emt}}^{\,g}(x,\omega) + \left (\mu +
\displaystyle \frac{i \omega}{c}\right )
\Phi_{\mathrm{emt}}^{\,g}(x,\omega) = \tilde{\gamma}(\omega)
\,c_{\mathrm{flr}}(x) \,
\Phi_{\mathrm{exc}}^{\,g}(x,\omega) \quad& \textrm{in} \, \Omega , \\
\vspace{0.2cm} \displaystyle \ell \frac{\partial
\Phi_{\mathrm{emt}}^{\,g}}{\partial \nu}(x,\omega) +\displaystyle
 \, \Phi_{\mathrm{emt}}^{\,g}(x,\omega) = 0 &
\textrm{on} \, \partial \Omega,
\end{array}
\right .
\end{equation}
where the source $g$ is in $L^2(\partial \Omega)$.

\subsection{Electrical model of a cell}

We apply at the boundary of our domain an electric field
$g_{\mathrm{ele}} \in L^2(\partial \Omega)$. We consider that
$\Omega\setminus \overline{C}$ and $C$ are homogeneous and
isotropic media with conductivity $1$. The thickness $\epsilon$ of
the cell membrane is supposed to be small. We denote by $\sigma$
the conductivity of the cell membrane. We assume that $\sigma \ll
1$ and $\beta >0$ to be given by $\beta = \sigma^{-1} \epsilon$,
see \cite{7}.

We can approximate the voltage potential $u$ within our medium by
the unique solution to the following problem
\cite{poig1,poig2,6,poig3,poig4}:

\begin{equation}\label{eq:u2}
\left\{
\begin{array}{ll}
\vspace{0.25 cm}\Delta u = 0 & \textrm{in}\, C \cup \Omega \setminus \overline{C},\\
\vspace{0.25 cm}\displaystyle\frac{\partial u}{\partial
\nu}\bigg|_{+} -
\frac{\partial u}{\partial \nu}\bigg|_{-} = 0 & \textrm{on}\, \partial C,\\
\vspace{0.25 cm}u\mid_{+} - u\mid_{-} = \beta \displaystyle
\frac{\partial u}{\partial \nu}
& \textrm{on} \,\partial C,\\
\displaystyle \frac{\partial u}{\partial \nu}\bigg|_{\partial
\Omega} = g_{\mathrm{ele}}, & \displaystyle \int_{\partial \Omega}
u = 0 .
\end{array}
\right.
\end{equation}

Since we have chosen the fluorescent indicators of the cell
membrane such that they respond linearly to the potential jump
across the membrane \cite{9}, we can express their concentration
as
\begin{equation}\label{eq:cf}
c_{\mathrm{flr}} = \delta \,[u]\big|_{\partial C} ,
\end{equation}
where $\delta$ is a constant \cite{9}.

\section{Forward problem} \label{sectforward}

The forward problem consists in determining $\Phi_{\mathrm
{emt}}|_{\partial \Omega}$, for a fixed applied electric field
$g_{\mathrm{ele}}$, a light excitation $g$ and a given cell $C$.
The optical parameters of the medium, $D$ and $\mu$, the speed of
light $c$, the extrapolation length $\ell$ and $\tilde{\gamma}$
are supposed to be known.

\subsection{Expression of $\Phi_{\mathrm{exc}}^{\,g}$}

Let $\Phi_{\mathrm{exc}}^{\,g}$ be the excitation light fluence in
$\Omega$, due to an excitation $g$ applied at its boundary
$\partial \Omega$. The function $\Phi_{\mathrm{exc}}^{\,g}$ is the
solution to the following problem:
\begin{equation}\label{eq:ex}
\left \{
\begin{array}{ll}
\vspace{0.2cm} - \Delta \Phi_{\mathrm{exc}}^{\,g}(y) + k^2 \, \Phi_{\mathrm{exc}}^{\,g}(y) = 0 \quad&\textrm{in} \, \Omega , \\
\vspace{0.2cm} \displaystyle \ell \frac{\partial
\Phi_{\mathrm{exc}}^{\,g}}{\partial \nu}(y) + \,
\Phi_{\mathrm{exc}}^{\,g}(y) = g & \textrm{on} \, \partial \Omega
,
\end{array}
\right .
\end{equation}
\noindent where $\vspace{0.1cm} k^2 = \displaystyle \frac{\mu + i
\omega/c}{D}$. Note that if $\ell=0$, then the Robin boundary
condition in (\ref{eq:ex}) should be replaced with the Dirichlet
boundary condition: $\Phi_{\mathrm{exc}}^{\,g}(y) = g \,
\textrm{on} \,
\partial \Omega$.

Let $\Gamma$ be the fundamental solution to $- \Delta + k^2$.
$\Gamma$ is (the exponentially decaying) solution to
\begin{equation}\label{eq:greenhz}
\forall \,y,z \in \mathbb R ^d,  \quad - \Delta_y \Gamma_z(y)+
k^2\, \Gamma_z(y) = \delta_z(y),
\end{equation}
where $\delta_z$ is the Dirac mass at $z$.

We know the explicit expression of $\Gamma_z(y)$ for all $y \neq z
\in \mathbb R^d$ \cite{1}:
\begin{displaymath}
\begin{array}{cl}
\vspace{0.25 cm} \Gamma_z(y) = \displaystyle \frac{i}{4}\,H_0^{(1)}(ik|y-z|) &\quad \textrm{if}\, d = 2, \\
 \Gamma_z(y) = \displaystyle \frac{e^{-k|y-z|}}{4 \pi |y-z|}&\quad \textrm{if} \,d = 3,
\end{array}
\end{displaymath}
\vspace{0.05 cm} where $H_0^{(1)}$ is the Hankel function of the
first kind of order $0$.

We introduce the single and double layer potentials of a function
$f \in L^2(\partial \Omega)$, for all $z \in \mathbb R^d \setminus
\partial \Omega$, \cite{1}:
\begin{displaymath}
\forall \, z \in \mathbb R^d, \quad  \mathcal{S}_{\Omega}[f](z) =
\int_{\partial \Omega} \Gamma_z(y) \,f(y)\, ds(y),
\end{displaymath}
\begin{displaymath}
 \forall \, z
\in \mathbb R^d \setminus
\partial \Omega, \quad \mathcal{D}_{\Omega}[f](z) = \int_{\partial \Omega} \displaystyle
\frac{\partial \Gamma_z(y)}{\partial \nu}\, f(y) \, ds(y) .
\end{displaymath}

\begin{lem} \label{prop311}
The double layer potential verifies, for all $f \in L^2(\partial
\Omega)$,
\begin{displaymath}
\begin{array}{cc}
\vspace{0.5cm}(-\Delta + k^2)\mathcal{D}_{\Omega}[f] = 0 & \qquad
\textrm{{\rm in}} \, \mathbb R^d
\setminus \partial \Omega,\\
\vspace{0.5cm}\displaystyle\frac{\partial} {\partial
\nu}\mathcal{D}_{\Omega}[f]|_+ = \displaystyle\frac{\partial}
{\partial \nu}\mathcal{D}_{\Omega}[f]|_-& \qquad \textrm{{\rm on}} \, \partial \Omega,\\
\mathcal{D}_{\Omega}[f]|_\pm =  \left(\mp \frac{1}{2} I +
\mathcal{K}_{\Omega}\right)[f]& \qquad \textrm{{\rm on}} \,
\partial \Omega,\\
\end{array}
\end{displaymath}
\noindent where $\mathcal{K}_{\Omega}: L^2(\partial \Omega)
\rightarrow L^2(\partial \Omega)$ is defined by
$$\forall \, z \in \partial \Omega, \quad
\mathcal{K}_{\Omega}[f](z)= \int_{\partial \Omega} \displaystyle \frac{\partial}{\partial y} \Gamma_z(y) f(y)\, ds(y).$$
\end{lem}

\begin{lem} \label{prop312}
Let $d=2,3$. The single layer potential verifies, for all $f \in
L^2(\partial \Omega)$,
\begin{displaymath}
\begin{array}{cc}
\vspace{0.5cm} (-\Delta + k^2)\mathcal{S}_{\Omega}[f] = 0& \qquad \textrm{{\rm in}}\, \mathbb R^d \setminus \partial \Omega,\\
 \mathcal{S}_{\Omega}[f]|_+ = \mathcal{S}_{\Omega}[f]|_- & \qquad \textrm{{\rm on}}\,\partial \Omega,
\end{array}
\end{displaymath}
 the single layer potential is therefore well defined on $\partial \Omega$, and hence on $\mathbb
 R^d$. Moreover,
$$\displaystyle\frac{\partial} {\partial \nu}\mathcal{S}_{\Omega}[f]|_\pm=  \left(\pm \frac{1}{2} I
+ \mathcal{K}_{\Omega}^*\right)[f] \qquad \textrm{{\rm on}} \,
\partial \Omega,$$ \noindent where $\mathcal{K}_{\Omega}^* : L^2(\partial \Omega) \rightarrow L^2(\partial \Omega)$ is the
$L^2$-adjoint of the operator $\mathcal{K}_{\Omega}$, {\it i.e.},
$$ \forall \, z \in \partial \Omega, \quad \mathcal{K}_{\Omega}^* [f](z)= \int_{\partial \Omega} \displaystyle \frac{\partial}{\partial z}
 \Gamma_z(y) f(y)\, ds(y).$$
\end{lem}

Let $G$ be the Green function of problem (\ref{eq:ex}), that is,
for all $z \in \Omega$, the unique solution to
\begin{equation} \label{eq:greenex}
\left \{
\begin{array}{ll}
\vspace{0.2cm} - \Delta_y G_z(y) + k^2 \, G_z(y) = \delta_z \quad& \textrm{in} \, \Omega ,\\
\vspace{0.2cm} \displaystyle \ell \frac{\partial G_z}{\partial
\nu}(y) + \,G_z(y) = 0 & \textrm{on} \,\partial \Omega .
\end{array}
\right .
\end{equation}

\begin{lem} \label{prop313}
The operator of kernel $G_z(y)$ is the solution operator for
problem (\ref{eq:ex}):
\begin{equation} \label{defprop313}
\forall z \in \Omega, \qquad \Phi_{\mathrm{exc}}^{\,g}(z) =
\frac{1}{\ell} \int_{\partial \Omega} G_z(y)g(y)\, ds(y).
\end{equation}
\end{lem}

\begin{proof}
Since $G_z$ and $\Phi_{\mathrm{exc}}^{\,g}$ are respectively the
solutions to problems (\ref{eq:greenex}) and (\ref{eq:ex}), we
have the equation:
\begin{displaymath}
\Phi_{\mathrm{exc}}^{\,g}(z)= \int_{\Omega} \bigg[(-\Delta_y
G_z(y) + k^2\,G_z(y))\Phi_{\mathrm{exc}}^{\,g}(y) - (-\Delta
\Phi_{\mathrm{exc}}^{\,g}(y) + k^2\, \Phi_{\mathrm{exc}}^{\,g}(y))
G_z(y) \bigg] \, dy  .
\end{displaymath}

Besides, we can apply Green's formula:
\begin{displaymath}
\begin{array}{l}
\vspace{0.4cm}\displaystyle \Phi_{\mathrm{exc}}^{\,g}(z)=
\displaystyle \int_{\Omega} \bigg[ -\Delta_y
G_z(y)\Phi_{\mathrm{exc}}^{\,g}(y)
+ \Delta \Phi_{\mathrm{exc}}^{\,g}(y)G_z(y) \bigg] \,dy \\
\qquad \qquad \qquad = \displaystyle \int_{\partial \Omega} \bigg[
- \displaystyle\frac{\partial G_z(y)}{\partial \nu}
\Phi_{\mathrm{exc}}^{\,g}(y) +\frac{\partial
\Phi_{\mathrm{exc}}^{\,g}(y)}{\partial \nu} G_z(y) \bigg] \,ds(y)
.
\end{array}
\end{displaymath}

 Using the boundary conditions that  $G_z$ and $\Phi_{\mathrm{exc}}^{\,g}$ verify, we then
 obtain that
\begin{displaymath}
\begin{array}{l}
\vspace{0.4cm}\displaystyle \Phi_{\mathrm{exc}}^{\,g}(z)=
\displaystyle \frac{1}{\ell} \int_{\partial \Omega} G_z(y) g(y)
\,ds(y) .
\end{array}
\end{displaymath}
\end{proof}

Thanks to the previous lemma, if we know $G_z$, we can calculate
the excitation light fluence for any source $g$. The following
result relates $G_z$, the Green function of our problem to
$\Gamma_z$, for which we have an explicit formula. It generalizes
\cite[Lemma 2.15]{lnm} to the Green function $G_z$.

\begin{prop}  \label{prop314} For $\vspace{0.15cm}z \in \Omega$ and $y \in \partial \Omega$,
\begin{equation} \label{eqprop314} \left(-\displaystyle \frac{I}{2} + \mathcal{K}_{\Omega} +
\frac{1}{\ell}\, \mathcal{S}_{\Omega}\right)[G_z](y) =
\Gamma_z(y).\end{equation} More precisely, for any simply
connected smooth domain $D$ compactly contained in $\Omega$, and
for any $h \in L^2(\partial D)$, we have for any $y \in
\partial \Omega$:
$$\int_{\partial D}\left(-\displaystyle \frac{I}{2}
+ \mathcal{K}_{\Omega} + \frac{1}{\ell}\,
\mathcal{S}_{\Omega}\right)[G_z](y)\,
 h(z)\, ds(z) = \int_{\partial D} \Gamma_z(y) \,h(z)\,ds(z).$$
\end{prop}

\begin{proof}
Let $f \in L_0^2(\partial \Omega)$, where $L_0^2(\partial \Omega)$
is the set of $L^2$ functions in $\Omega$ of mean zero. For $z \in
\Omega$ and $y \in
\partial \Omega$, we define
\begin{displaymath}
u(z) := \int_{\partial \Omega} \left(-\displaystyle \frac{I}{2} +
\mathcal{K}_{\Omega} + \frac{1}{\ell}\,
\mathcal{S}_{\Omega}\right) [G_z](y)f(y)\,ds(y).
\end{displaymath}

By introducing the adjoint operator, we obtain
\begin{displaymath}
\begin{array}{ll}
\vspace{0.3cm} u(z) & = \displaystyle \int_{\partial
\Omega}G_z(y)\left (-\displaystyle \frac{I}{2} +
\mathcal{K}_{\Omega}^*+ \frac{1}{\ell}
\,\mathcal{S}_{\Omega}\right )[f](y)\,ds(y).
\end{array}
\end{displaymath}

By Lemma \ref{prop313}, $u$ is then solution to the problem:
\begin{equation} \label{eq:u}
\left \{
\begin{array}{ll}
\vspace{0.2cm} - \Delta u(y) + k^2 \, u(y) = 0 \quad& \textrm{in} \, \Omega, \\
\vspace{0.2cm} \displaystyle \frac{\partial u}{\partial \nu}(y)
+\frac{1}{\ell} \, u(y) = \left (-\displaystyle \frac{I}{2} +
\mathcal{K}_{\Omega}^*+ \frac{1}{\ell}\,\mathcal{S}_{\Omega}\right
)[f](y) & \textrm{on} \,\partial \Omega .
\end{array}
\right .
\end{equation}

We know that $\mathcal{S}_{\Omega}[f]$ is although solution to the
problem (\ref{eq:u}), thanks to Lemma \ref{prop312}. The equation
$(- \Delta + k^2) p = 0 $ in $\Omega$ with the Robin boundary
condition, $\partial p/\partial \nu + l p =0$, admits a unique
solution, provided that $l>0$. Therefore, we have
$$\forall z \in \Omega,  \qquad u(z) = \mathcal{S}_{\Omega}[f](z).$$
Since $f$ is arbitrary, we have therefore proved the first part of
our proposition.

Let $h \in L^2(\partial D)$. By multiplying the last equality by
$h$ and integrating on $\partial D$, we obtain
$$\int_{\partial \Omega}\!\int_{\partial D}\!\!\left(-\displaystyle \frac{I}{2}
 + \mathcal{K}_{\Omega} + \frac{1}{\ell} \, \mathcal{S}_{\Omega}\right)\!\![G_z](y)h(z)f(y)\, ds(z) \, ds(y)=
 \!\!\int_{\partial \Omega} \!\int_{\partial D} \!\!\Gamma_z(y) h(z)f(y)\, ds(z)\, ds(y),$$
which completes the proof.
\end{proof}

According to the previous proposition, the knowledge of
$\vspace{0.5mm}G_z$, and therefore of $\Phi_{\mathrm{exc}}^{\,g}$,
requires the inversion of the operator:
\begin{equation}
\label{eq:operator}
 -\displaystyle \frac{I}{2} + \mathcal{K}_{\Omega} + \frac{1}{\ell}\, \mathcal{S}_{\Omega} : L^2(\partial \Omega)
 \rightarrow L^2(\partial \Omega).
 \end{equation}
 In the case of circular domains, we can exhibit an explicit formula of the inverse operator.

\paragraph{Explicit calculation of $G_z$ for a circular domain:}

We assume that the dimension is two and $\Omega$ is the unit disk.
In terms of polar coordinates, the fundamental solution $\Gamma_z$
to $- \Delta + k^2$ has the expression:
\begin{displaymath}
 \forall y\,(r,\theta) \in \overline{\Omega},\, \forall z\,(R,\phi) \in \overline{\Omega},  \qquad
 \Gamma_z(y) = \displaystyle \frac{i}{4}\,H_0^{(1)}(ik|re^{i\theta}-Re^{i\phi}|).
\end{displaymath}

Graf's formula \cite[Formula (9.1.79)]{12} gives us the following
decomposition of $\Gamma_z$:

\begin{displaymath}
H_0^{(1)}(ik|re^{i\theta}-Re^{i\phi}|) = \displaystyle \sum_{m \in
\mathbb Z} H_{m}^{(1)}(ikr) J_m(ikR) e^{im(\theta - \phi)}, \quad
r>R,
\end{displaymath}
with $H_m^{(1)}$ and $J_m$ being respectively the Hankel and
Bessel functions of the first kind of order $m$.

For all $g \in L^2(]0, 2\pi[)$, we introduce the Fourier
coefficients:
\begin{displaymath}
 \forall m \in \mathbb Z,  \quad  \hat{g}(m) = \displaystyle \frac{1}{2\pi} \int_0^{2\pi} g(\phi) e^{-im\phi}
 d\phi,
\end{displaymath}
and have then
\begin{displaymath}
 g(\phi) =
\sum_{m=-\infty}^{\infty} \hat{g}(m) e^{im\phi} \quad \mbox{in }
L^2.
\end{displaymath}
Let $D$ be the disk with radius $R$ and center $0$. For
$y(r,\theta) \in \overline{\Omega}$,
\begin{displaymath}
\begin{array}{rl}
\vspace{0.25cm}\mathcal{S}_{D}[g](y) &=\displaystyle \frac{iR}{4}
\int_{0}^{2\pi} H_0^{(1)}(ik|re^{i\theta}-Re^{i\phi}|) g(\phi) d\phi , \\
\vspace{0.25cm}&=\displaystyle \frac{iR}{4} \sum_{m=-\infty}^{\infty} H_{m}^{(1)}(ikr) J_m(ikR) e^{im\theta} \int_0^{2\pi} g(\phi) e^{-im\phi} d\phi\\
& = \displaystyle \frac{iR\pi}{2} \sum_{m=-\infty}^{\infty}
H_{m}^{(1)}(ikr) J_m(ikR) \hat{g}(m) e^{im\theta} .
\end{array}
\end{displaymath}
For $y\,(1,\theta)\in \partial \Omega$, we therefore obtain
\begin{displaymath}
\begin{array}{rl}
\vspace{0.25cm}\mathcal{S}_{D}[g](y) &= \ds
\sum_{m=-\infty}^{\infty} \widehat{\mathcal S_D}(m) \hat{g}(m)
e^{im\theta} ,
\end{array}
\end{displaymath}
 with
\begin{displaymath}
\forall m \in \mathbb Z,  \quad  \widehat{\mathcal S_D}(m) =
\displaystyle \frac{iR\pi}{2} H_{m}^{(1)}(ik) J_m(ikR),
\end{displaymath}
and analogously,
\begin{displaymath}
\begin{array}{rl}
\vspace{0.25cm}\mathcal{S}_{\Omega}[g](y) &=\ds
\sum_{m=-\infty}^{\infty} \widehat{\mathcal S_\Omega}(m)
\hat{g}(m) e^{im\theta} ,
\end{array}
\end{displaymath}
with
\begin{displaymath}
\forall m \in \mathbb Z, \quad  \widehat{\mathcal S_{\Omega}}(m) =
\displaystyle \frac{i\pi}{2} H_{m}^{(1)}(ik) J_m(ik).
\end{displaymath}
We can prove, in a similar way, that
\begin{displaymath}
\begin{array}{rl}
\vspace{0.25cm} \mathcal K_{\Omega}[g](y) &=\ds
\sum_{m=-\infty}^{\infty} \widehat{\mathcal K_{\Omega}}(m)
\hat{g}(m) e^{im\theta},
\end{array}
\end{displaymath}
with
\begin{displaymath}
\forall m \in \mathbb Z, \quad  \widehat{\mathcal K_{\Omega}}(m) =
\displaystyle \frac{-k\pi}{2} H_{m}^{(1)}(ik) J_m'(ik).
\end{displaymath}
Using Proposition \ref{prop314}, we can express the Fourier
coefficients of the operator with kernel $G_z(y)$ for all
$z(R,\theta) \in
\partial D$ defined by
\begin{displaymath}
\int_{\partial \Omega} G_z(y) g(y)\, ds(y) =
\sum_{m=-\infty}^{\infty}  \widehat{G}(m) \hat{g}(m) e^{im\theta},
\end{displaymath} as follows:
\begin{displaymath}
\forall m \in \mathbb Z, \quad  \widehat{G}(m) = \displaystyle
\frac { \widehat{\mathcal S_D}(m)}{\widehat{\mathcal
K_{\Omega}}(m) + \frac{1}{\ell} \widehat{\mathcal S_{\Omega}}(m)},
\end{displaymath}
that is,
\begin{displaymath}
\forall m \in \mathbb Z, \quad  \widehat{G}(m) = \displaystyle
\frac {J_{m}(ikR)  }{ik J_{m}'(ik) + \frac{1}{\ell} J_{m}(ik)}.
\end{displaymath}
Moreover, the function $\Phi_{\mathrm{exc}}^{\,g}$ defined by
(\ref{defprop313}) can be written as
\begin{equation} \label{eq:disk}
\Phi_{\mathrm{exc}}^{\,g}(R, \theta) = \sum_{m=-\infty}^{\infty}
\displaystyle \frac {J_{m}(ikR)  }{ik \ell J_{m}'(ik) + J_{m}(ik)}
\hat{g}(m) e^{im\theta}.
\end{equation}

When $\Omega$ is approximated by the unit disk, we have shown that
we can easily invert our operator (\ref{eq:operator}) and obtain
an explicit formula of our Green's function $G_z$. We can then
calculate the excitation light fluence, for any source $g$, in
this particular case. The same result holds for the unit sphere,
see Appendix \ref{appendixA}.

\subsection{Expression of $c_{\mathrm{flr}}$}

Recall that the concentration of fluorophores $c_{\mathrm{flr}}$
can be expressed as
\begin{displaymath}
c_{\mathrm{flr}} = \delta \,[u]\big|_{\partial C} ,
\end{displaymath}
where $\delta$ is a constant and $u$,the voltage potential in our
domain, satisfies (\ref{eq:u2}).

Let $L^2_0(\partial C):= \{ \Psi \in L^2(\partial C) :
\int_{\partial C} \Psi =0\}$. Let $\Gamma^{(0)}$ be the
fundamental solution to $\Delta$ in $\mathbb R^d$:
\begin{equation} \label{gamma0}
\Gamma^{(0)}(x) :=\left\{ \begin{array}{ll} \displaystyle
 \frac{1}{2\pi} \log |x|,&
\qquad d=2,\\
\nm \displaystyle  - \frac{1}{4\pi |x|},& \qquad d=3.
\end{array} \right.
\end{equation}

Analogously to Section \ref{sectforward}, we introduce the layer
potentials, $\mathcal{S}^{(0)}_C, \mathcal{S}^{(0)}_\Omega,
\mathcal{D}^{(0)}_C, \mathcal{D}^{(0)}_\Omega,
\mathcal{K}^{(0)}_C,$ and $(\mathcal{K}^{(0)}_C)^*$ associated
with $\Gamma^{(0)}$. The following proposition from \cite{6} gives
us a representation formula for the voltage potential in $\Omega$.

\begin{prop} \label{prop421}
There exists at most one solution $u$ to the problem (\ref{eq:u2})
and it satisfies the following representation formula:

\begin{equation}\label{eq:decu}
\forall x \in \Omega, \quad u(x) = H(x)+
\mathcal{D}^{(0)}_C[\Psi](x),
\end{equation}

\noindent where the harmonic function $H$ is given by

\begin{equation}\label{eq:decH}
\forall x \in \mathbb R^2 \setminus \partial \Omega,  \quad H(x) =
- \mathcal{S}^{(0)}_{\Omega}[g_{\mathrm{ele}}](x) +
\mathcal{D}^{(0)}_{\Omega}[u|_{\partial \Omega}](x),
\end{equation}

\noindent and $\Psi \in L^2_0(\partial C)$ satisfies the integral
equation:

\begin{equation}\label{eq:psi}
\Psi + \beta \displaystyle \frac{\partial
\mathcal{D}^{(0)}_C[\Psi]}{\partial \nu} = - \beta \frac{\partial
H}{\partial \nu} \quad \textrm{on} \, \partial C.
\end{equation}

\noindent The decomposition in (\ref{eq:decu}) is unique.
Furthermore, the following identity holds:

$$
\forall x \in \mathbb R^2 \setminus \overline{\Omega},  \quad u(x)
= H(x)+ \mathcal{D}^{(0)}_C[\Psi](x) = 0.
$$
\end{prop}

Since the normal derivative of the layer potential is continuous
across its boundary, the representation formula (\ref{eq:decu})
gives us an expression for $\displaystyle \frac{\partial
u}{\partial \nu}\big |_{\partial C}$, and hence  for
$c_{\mathrm{flr}}$ thanks to (\ref{eq:u2}) and (\ref{eq:cf}). For
a given applied electric field $g_{\mathrm{ele}}$ and cell $C$,
one can therefore compute the fluorophore concentration
$c_{\mathrm{flr}}$ on $\partial C$.

\subsection{Expression of $\Phi_{\mathrm{emt}}^{\,g}$}

The emitted light fluence $\Phi_{\mathrm{emt}}^{\,g}$ due to an
excitation $g$ is the solution to the following problem:
\begin{equation}\label{eq:em}
\left \{
\begin{array}{ll}
\vspace{0.2cm} - \Delta \Phi_{\mathrm{emt}}^{\,g}(y) + k^2 \,
\Phi_{\mathrm{emt}}^{\,g}(y) = \displaystyle
\frac{\tilde{\gamma}}{D} \,c_{\mathrm{flr}}(y) \, \Phi_{\mathrm{exc}}^{\,g}(y) \quad& \textrm{in} \, \Omega ,\\
\vspace{0.2cm} \displaystyle \ell \frac{\partial
\Phi_{\mathrm{emt}}^{\,g}}{\partial \nu}(y) + \,
\Phi_{\mathrm{emt}}^{\,g}(y) = 0 & \textrm{on} \, \partial \Omega,
\end{array}
\right .
\end{equation}
where $\Phi_{\mathrm{exc}}^{\,g}$ is the excitation light fluence
launched by the source $g$ in $\Omega$.

The measured quantity on $\partial \Omega$ is
$$I^{\, g}_{\mathrm{emt}} = -D \frac{\partial
\Phi_{\mathrm{emt}}^{\,g}}{\partial \nu}\bigg|_{\partial
\Omega},$$ which is the outgoing light intensity determined from
Fick's law. It is worth mentioning that, in our coupled diffusion
equations model, if $\ell\neq 0$, then knowing $
\Phi_{\mathrm{emt}}^{\,g}$ or ${\partial
\Phi_{\mathrm{emt}}^{\,g}}/{\partial \nu}$ on $\partial \Omega$ is
mathematically the same.

\begin{prop} \label{prop331}
The emitted light fluence $\Phi_{\mathrm{emt}}^{\,g}$ can be
expressed as a function of $G_z$ and $\Phi_{\mathrm{exc}}^{\,g}$
as follows:
$$\forall z \in \overline{\Omega},  \qquad \Phi_{\mathrm{emt}}^{\,g}(z)
= \int_{\partial C} \frac{\tilde{\gamma}}{D}  \,G_z(y)
\,c_{\mathrm{flr}}(y) \, \Phi_{\mathrm{exc}}^{\,g}(y) \,ds(y),$$
where $\partial C$ is the cell membrane.
\end{prop}

\begin{proof}
Since $G$ and $\Phi_{\mathrm{emt}}^{\,g}$ are the solutions to the
problems (\ref{eq:greenex}) and (\ref{eq:em}), we have
\begin{displaymath}
\begin{array}{l}
\vspace{2mm} \ds \Phi_{\mathrm{emt}}^{\,g}(z) - \int_{\Omega}
\frac{\tilde{\gamma}}{D}  \, G_z(y) \,c_{\mathrm{flr}}(y) \,
  \Phi_{\mathrm{exc}}^{\,g}(y)
 \,dy  = \int_{\Omega} \bigg[(-\Delta_y G_z(y) +
k^2\,G_z(y))\Phi_{\mathrm{emt}}^{\,g}(y) \\ \nm \qquad \ds -
G_z(y) (-\Delta \Phi_{\mathrm{emt}}^{\,g}(y) + k^2\,
\Phi_{\mathrm{emt}}^{\,g}(y)) \bigg] \, dy.
\end{array}
\end{displaymath}
Besides, we can apply Green's formula:
\begin{displaymath}
\begin{array}{l}
\vspace{0.4cm} \ds \Phi_{\mathrm{emt}}^{\,g}(z) - \int_{\Omega}
\frac{\tilde{\gamma}}{D} \, G_z(y) \,c_{\mathrm{flr}}(y) \,
  \Phi_{\mathrm{exc}}^{\,g}(y)
 \,dy
= \ds \int_{\partial \Omega} \bigg[ - \displaystyle\frac{\partial
G_z(y)}{\partial \nu} \Phi_{\mathrm{emt}}^{\,g}(y) + G_z(y)
\frac{\partial \Phi_{\mathrm{emt}}^{\,g}(y)}{\partial \nu}
\bigg]\,ds(y) .
\end{array}
\end{displaymath}
 Using the boundary conditions that  $G_z$ and $\Phi_{\mathrm{emt}}^{\,g}$ verify, we then obtain
\begin{displaymath}
\begin{array}{lll}
 \ds \Phi_{\mathrm{emt}}^{\,g}(z) - \int_{\Omega}
\frac{\tilde{\gamma}}{D}  \, G_z(y) \,c_{\mathrm{flr}}(y) \,
  \Phi_{\mathrm{exc}}^{\,g}(y)
 \,dy
&=& \ds \int_{\partial \Omega} \bigg[\frac{1}{\ell} G_z(y)
\Phi_{\mathrm{emt}}^{\,g}(y)
+ G_z(y)\displaystyle\frac{\partial \Phi_{\mathrm{emt}}^{\,g}(y)}{\partial \nu} \bigg]\, ds(y),\\
&=& 0 .
\end{array}
\end{displaymath}
Since the concentration of the fluorophores is zero except on
$\partial C$, we get finally the formula:
$$\forall z \in \overline{\Omega},  \qquad \Phi_{\mathrm{emt}}^{\,g}(z) = \int_{\partial C} \frac{\tilde{\gamma}}{D}
\,G_z(y) \,c_{\mathrm{flr}}(y) \, \Phi_{\mathrm{exc}}^{\,g}(y)
\,ds(y).$$
\end{proof}

By combining the results of the first section and of this last
section, for a given concentration of fluorophore
$c_{\mathrm{flr}}$ and an excitation $g$, we can express
$\Phi_{\mathrm{emt}}^{\,g}$, at any point of $\overline{\Omega}$,
and in particular on $\partial \Omega$. Moreover, section 3.2
gives us a unique formula for the
 fluorophore concentration for given $g_{\mathrm{ele}}$ and $C$. If we couple these two
 formulas, we solve our forward problem.

\section{Inverse problem}

The shape and position of the cell $C$ are now considered to be
unknown. We illuminate our domain with a light source $g$ and
apply an electric field $g_{\mathrm{ele}}$ at its boundary. We
measure an outgoing light intensity $I^{\, g}_{\mathrm{emt}}$. Our
goal is to reconstruct the concentration of fluorophore
$c_{\mathrm{flr}}$. We will thus have an image of the membrane
potential changes and hence locate the cell. In this section we
consider only the two-dimensional case. We start with the
reconstruction of the cell membrane $\partial C$ in the case where
it is assumed to be a perturbation of a disk. We derive analytical
formulas for the resolving power of the proposed imaging method in
two different regimes. Then we extend our results to arbitrary
shapes. In three dimensions, similar results hold and analytical
formulas for the resolving power of the imaging method can be
derived for $\partial C$ being a perturbation of a sphere.

\subsection{Problem Formulation}

The excitation light fluence, $\Phi_{\mathrm{exc}}^{\,f}$, due to
a source $f \in L^2(\partial \Omega)$, is the solution to
\begin{displaymath}
\left \{
\begin{array}{ll}
\vspace{0.2cm} - \Delta \Phi_{\mathrm{exc}}^{\,f}(y) + k^2 \, \Phi_{\mathrm{exc}}^{\,f}(y) = 0 \quad& \textrm{in}\, \Omega ,\\
\vspace{0.2cm} \displaystyle \ell \, \frac{\partial
\Phi_{\mathrm{exc}}^{\,f}}{\partial \nu}(y) + \,
\Phi_{\mathrm{exc}}^{\,f}(y) = f & \textrm{on} \, \partial \Omega
.
\end{array}
\right .
\end{displaymath}

\noindent We denote by $\Phi_{\mathrm{exc}}^{\,g}$ the excitation
light fluence due to an excitation $g \in L^2(\partial \Omega)$.
The emitted light fluence, $\Phi_{\mathrm{emt}}^{\,g}$, due to the
excitation of the fluorophores by $\Phi_{\mathrm{exc}}^{\,g}$,
verifies
\begin{displaymath}
\left \{
\begin{array}{ll}
\vspace{0.2cm} \displaystyle - \Delta \Phi_{\mathrm{emt}}^{\,g}(y)
+ k^2 \,
\Phi_{\mathrm{emt}}^{\,g}(y) = \frac{\tilde{\gamma}}{D}  \,c_{\mathrm{flr}}(y) \, \Phi_{\mathrm{exc}}^{\,g}(y) \quad& \textrm{in} \, \Omega ,\\
\vspace{0.2cm} \displaystyle \ell\, \frac{\partial
\Phi_{\mathrm{emt}}^{\,g}}{\partial \nu}(y) + \,
\Phi_{\mathrm{emt}}^{\,g}(y) = 0 & \textrm{on} \,\partial \Omega .
\end{array}
\right .
\end{displaymath}
By multiplying the last equation by $\Phi_{\mathrm{exc}}^{\,f}$
and integrating on our domain $\Omega$, we obtain the following
formula:
\begin{displaymath}
\int_{\Omega} \frac{\tilde{\gamma}}{D}  \,c_{\mathrm{flr}}(y) \,
\Phi_{\mathrm{exc}}^{\,g}(y) \Phi_{\mathrm{exc}}^{\,f} (y) \,dy =
\int_{\Omega} \bigg[- \Delta \Phi_{\mathrm{emt}}^{\,g}(y)
\Phi_{\mathrm{exc}}^{\,f}(y) + k^2 \,\Phi_{\mathrm{emt}}^{\,g}(y)
\Phi_{\mathrm{exc}}^{\,f}(y) \bigg]\, dy.
\end{displaymath}
From the first equation, we know that in $\Omega$:
\begin{displaymath}
k^2  \,\Phi_{\mathrm{exc}}^{\,f} \Phi_{\mathrm{emt}}^{\,g} =
\Delta \Phi_{\mathrm{exc}}^{\,f} \Phi_{\mathrm{emt}}^{\,g}.
\end{displaymath}
Hence, we have
\begin{displaymath}
\int_{\Omega} \frac{\tilde{\gamma}}{D}  \,c_{\mathrm{flr}}(y) \,
\Phi_{\mathrm{exc}}^{\,g}(y) \Phi_{\mathrm{exc}}^{\,f} (y) \,dy =
\int_{\Omega}\bigg[ - \Delta \Phi_{\mathrm{emt}}^{\,g}(y)
\Phi_{\mathrm{exc}}^{\,f}(y) + \Delta \Phi_{\mathrm{exc}}^{\,f}(y)
\Phi_{\mathrm{emt}}^{\,g}(y) \bigg]\, dy.
\end{displaymath}
Green's formula gives us
\begin{displaymath}
\int_{\Omega} \frac{\tilde{\gamma}}{D}  \,c_{\mathrm{flr}}(y) \,
\Phi_{\mathrm{exc}}^{\,g}(y)\Phi_{\mathrm{exc}}^{\,f} (y) \, dy =
\int_{\partial \Omega} \bigg[-\frac{
\partial \Phi_{\mathrm{emt}}^{\,g}}{\partial \nu}(y)
\Phi_{\mathrm{exc}}^{\,f}(y) + \frac{\partial
\Phi_{\mathrm{exc}}^{\,f}}{\partial \nu}(y)
\Phi_{\mathrm{emt}}^{\,g}(y) \bigg] \, ds(y).
\end{displaymath}
We use the boundary conditions of our two equations and obtain
that
\begin{displaymath}
\int_{\Omega} \frac{\tilde{\gamma}}{D}  \,c_{\mathrm{flr}}(y) \,
\Phi_{\mathrm{exc}}^{\,g}(y) \Phi_{\mathrm{exc}}^{\,f} (y) \,dy=
\frac{1}{\ell} \int_{\partial \Omega} f(y)  \,
\Phi_{\mathrm{emt}}^{\,g}(y) \, ds(y).
\end{displaymath}

The concentration of the fluorophores is zero except on $\partial
C$, so we get finally the following proposition.

\begin{prop}\label{prop411}
Let $f$ and $g$ be in $L^2(\partial \Omega)$. The outgoing light
intensity  $I_{\mathrm{emt}}^{\,g} = -D \frac{\partial
\Phi_{\mathrm{emt}}^{\,g}}{\partial \nu}$ measured on $\partial
\Omega$, satisfies the formula:
\begin{equation}\label{eq:pbinv}
\int_{\partial C} \tilde{\gamma} \,c_{\mathrm{flr}}(y) \,
\Phi_{\mathrm{exc}}^{\,g}(y) \Phi_{\mathrm{exc}}^{\,f} (y) \,ds(y)
= \int_{\partial \Omega} f(y)  \, I_{\mathrm{emt}}^{\,g}(y) \,
ds(y).
\end{equation}
This formula also holds for $\ell =0$.
\end{prop}

For two chosen excitations $\vspace{0.2cm} f,g \in L^2(\partial
\Omega)$ and a measured outgoing light intensity
$I_{\mathrm{emt}}^{\,g}$, we can compute the integral
$\vspace{0.2cm} \displaystyle \int_{\partial \Omega} f(y)  \,
I_{\mathrm{emt}}^{\,g}(y) \, ds(y)$, and hence, thanks to the last
formula, know $\displaystyle \int_{\partial C} \tilde{\gamma}
\,c_{\mathrm{flr}}(y) \, \Phi_{\mathrm{exc}}^{\,g}(y)
\Phi_{\mathrm{exc}}^{\,f} (y) \,ds(y)$. Recall that the constant
$\tilde{\gamma}$ is assumed to be known. Then, if we choose
properly $f$ and $g$, we will be able to reconstruct
$c_{\mathrm{flr}} \,\mathbb{1}_{\partial C}$, and therefore to
image the cell membrane $\partial C$.

\subsection{Reconstruction of the  cell membrane: case of a perturbed disk}
\label{subsectcell} We consider a circular cell $C$ with radius
$R$. We choose to excite our medium with a source given by
$$f_n(\phi) = E_n\, e^{in\phi},$$
 for $n \in \mathbb{Z}, \phi \in [0,2\pi]$ and $E_n := ik \ell J_{n}'(ik) +  J_n(ik)$. It gives us,
 thanks to formula (\ref{eq:disk}), the excitation
 light fluence $\Phi_{\mathrm{exc}}^{n}$:

\begin{displaymath}
\forall \theta \in [0,2\pi], \quad
\Phi_{\mathrm{exc}}^{n}(R,\theta) = J_n(ikR) e^{-in\theta}.
\end{displaymath}

Let $\Phi^n_{\mathrm{emt}}$ be the emitted light fluence and let
$I^n_{\mathrm{emt}} = -D \frac{\partial
\Phi^n_{\mathrm{emt}}}{\partial \nu} |_{\partial \Omega}$ be the
outgoing light intensity measured at ${\partial \Omega}$ when the
cell occupies $C$ and the source $f_n$ is applied at $\partial
\Omega$. It follows from (\ref{eq:pbinv}) that
\begin{equation}\label{eq:iepsilonm}
\int_{\partial C} \tilde{\gamma} c_{\mathrm{flr}}(\theta)
\Phi_{\mathrm{exc}}^n(R,\theta)\Phi_{\mathrm{exc}}^m(R,\theta) R
d\theta = {2 \pi} \,E_m \, \widehat{I_{\mathrm{emt}}^n}(m) .
\end{equation}
Besides, we also have
\begin{displaymath}
\int_{\partial C} \tilde{\gamma} c_{\mathrm{flr}}(\theta)
\Phi_{\mathrm{exc}}^n(R,\theta)\Phi_{\mathrm{exc}}^m(R,\theta) R
d\theta = 2 \pi \tilde{\gamma} R\, J_n(ikR) J_m(ikR)\,
\widehat{c_{\mathrm{flr}}}(n+m) .
\end{displaymath}

Let $C_{\epsilon}$ be an $\epsilon$-perturbation of $C$, {\it
i.e.},  there is $h \in \mathcal{C}^2([0,2\pi])$, such that
$\partial C_{\epsilon}$ is given by
\begin{displaymath}
\partial C_{\epsilon} = \left\{\displaystyle \tilde{x}; \tilde{x}(\theta) = (R +
\epsilon h(\theta)) {e_r}, \theta \in [0,2\pi]\right\},
\end{displaymath}
with $({e_r}, {e_{\theta}})$ being the  basis of polar
coordinates.

Our goal is to reconstruct the shape deformation $h$ of our cell.
Let $\Phi^n_{\mathrm{emt},\epsilon}$ be the emitted light fluence
and let $I^n_{\mathrm{emt},\epsilon} = -D \frac{\partial
\Phi^n_{\mathrm{emt},\epsilon}}{\partial \nu}|_{\partial \Omega}$
be the outgoing light intensity measured at the boundary of our
domain $\Omega$ when the cell occupies $C_{\epsilon}$ and the
source $f_n$ is applied at $\partial \Omega$. Again, it follows
from (\ref{eq:pbinv}) that
\begin{equation}\label{eq:im}
\int_{\partial C_{\epsilon}} \tilde{\gamma}
\widetilde{c_{\mathrm{flr}}}(x)
\Phi_{\mathrm{exc}}^n(x)\Phi_{\mathrm{exc}}^m(x) \, ds(x) = {2
\pi} \,E_m \,\widehat{I_{\mathrm{emt}, \epsilon}^n}(m) .
\end{equation}
On the other hand, we have
\begin{equation}\label{eq:int}
\int_{\partial C_{\epsilon}} \tilde{\gamma}
\widetilde{c_{\mathrm{flr}}}(x)
\Phi_{\mathrm{exc}}^n(x)\Phi_{\mathrm{exc}}^m(x) \, ds(x) =
\int_{\partial C} \tilde{\gamma}
\widetilde{c_{\mathrm{flr}}}(\tilde{x})\, J_n(ik\tilde{R}(\theta))
J_m(ik\tilde{R}(\theta))\,
e^{-i(n+m)\theta}ds_{\epsilon}(\tilde{x}),
\end{equation}
where $\tilde{R}(\theta) = R + \epsilon h(\theta)$ and
$\widetilde{c_{\mathrm{flr}}}$ is the concentration of
fluorophores on the deformed cell membrane $\partial
C_{\epsilon}$.

We want to compute the first order approximation of our integral
(\ref{eq:int}). Taylor-Lagrange's theorem gives us the following
expansions, for all $N \in \mathbb N$:
\begin{equation}\label{eq:dechn}
\begin{array}{c}
\vspace{0.1 cm} J_m(ik\tilde{R}) = \displaystyle
\sum_{p=0}^{N}\frac{(ik\epsilon h(\theta))^p}{p!}
J^{(p)}_m(ikR) + o(\epsilon^N),\\
J_n(ik\tilde{R}) = \displaystyle \sum_{p=0}^{N}\frac{(ik\epsilon
h(\theta))^p}{p!} J^{(p)}_n(ikR) + o(\epsilon^N).
\end{array}
\end{equation}
In particular, at first order,
\begin{equation}\label{eq:dechnf}
\begin{array}{cccc}
\vspace{0.1 cm} J_m(ik\tilde{R}) &= J_m(ikR) &+ \, \epsilon \, ik \,h(\theta) J_m'(ikR) &+ \,o(\epsilon),\\
J_n(ik\tilde{R}) &= J_n(ikR) &+\, \epsilon \, ik  \, h(\theta)
J_n'(ikR) &+\, o(\epsilon).
\end{array}
\end{equation}
We can easily get an expansion for the length element
$ds_{\epsilon}(\tilde{y})$, for $\tilde{y} \in \partial
C_{\epsilon}$:
\begin{equation}\label{eq:sigma}
ds_{\epsilon}(\tilde{y}) = |\tilde{x'}(\theta)| d\theta =
\left((R+ \epsilon h(\theta))^2 + (\epsilon
h'(\theta))^2\right)^{\frac{1}{2}} d\theta = \sum_{n=0}^{\infty}
\epsilon^n\sigma^{(n)}(\theta)d\theta,
\end{equation}
where $\sigma^{(n)}$ are functions bounded independently of $n$
and, at first order, we have
\begin{equation}\label{eq:sigma1}
ds_{\epsilon}(\tilde{y}) = R d\theta + \epsilon h(\theta) d\theta
+ o(\epsilon).
\end{equation}

\subsubsection{High-order terms in the expansion of $\widetilde{c_{\mathrm{flr}}}$}

We denote $u_{\epsilon}$ (resp. $u$) the voltage potential in our
medium, when the cell occupies $C_{\epsilon}$ (resp. $C$). We
assume, thanks to (\ref{eq:cf}), that our concentration of
fluorophores $\widetilde{c_{\mathrm{flr}}}$ (resp.
$c_{\mathrm{flr}}$) on $\partial C_{\epsilon}$ (resp. $\partial
C$) is given by
\begin{displaymath}
\begin{array}{ll}
\vspace{0.1 cm}  &\widetilde{c_{\mathrm{flr}}} = \delta \,[u_{\epsilon}]\big|_{\partial C_{\epsilon}}\\
\textrm{resp.} \qquad & c_{\mathrm{flr}} = \delta
\,[u]\big|_{\partial C} .
\end{array}
\end{displaymath}

To find the first order term in the expansion of
$\widetilde{c_{\mathrm{flr}}}$, we must therefore expand at first
$u_{\epsilon}$. Similar problems have been considered in
\cite{elisa1,elisa2}. Nevertheless, our derivations, based on a
layer potential technique, differ significantly from those in
\cite{elisa1,elisa2}.

We know, from Proposition \ref{prop421}, that $u_{\epsilon}$
(resp. $u$) admits the following representation formula:

\begin{displaymath}
\begin{array}{rccccc}
\vspace{0.15cm} \forall x \in \Omega,  &\quad u_{\epsilon}(x) &=& H_{\epsilon}(x) &+&
\mathcal{D}^{(0)}_{C_{\epsilon}}[\Psi_{\epsilon}](x)\\
\textrm{resp.}\quad \forall x \in \Omega, & \quad u(x) &=& H(x)
&+& \mathcal{D}^{(0)}_C[\Psi](x),
\end{array}
\end{displaymath}

\noindent where the harmonic function $H_{\epsilon}$ (resp. $H$)
is given by

\begin{displaymath}
\begin{array}{rccccc}
\vspace{0.15cm} \forall x \in \mathbb R^2 \setminus \partial
\Omega, & \quad H_{\epsilon}(x) &=& -
 \mathcal{S}^{(0)}_{\Omega}[g_{\mathrm{ele}}](x) &+&  \mathcal{D}^{(0)}_{\Omega}[u_{\epsilon}|_{\partial \Omega}](x)\\
\textrm{resp.}\quad \forall x \in \mathbb R^2 \setminus \partial
\Omega, & \quad H(x) &=& -
\mathcal{S}^{(0)}_{\Omega}[g_{\mathrm{ele}}](x) &+&
\mathcal{D}^{(0)}_{\Omega}[u|_{\partial \Omega}](x),
\end{array}
\end{displaymath}
and $\Psi_{\epsilon} \in L^2_0(\partial C_{\epsilon})$ (resp.
$\Psi \in L^2_0(\partial C)$) satisfies the integral equation:
\begin{equation} \label{psieps}
\begin{array}{rccccc}
\vspace{0.15cm}  \Psi_{\epsilon} &+& \beta \displaystyle \frac{\partial
\mathcal{D}^{(0)}_{C_{\epsilon}}[\Psi_{\epsilon}]}{\partial \tilde{\nu}} &=& - \beta \displaystyle \frac{\partial H_{\epsilon}}{\partial \tilde{\nu}} \quad &\textrm{on} \, \partial C_{\epsilon}\\
\end{array}
\end{equation}
\begin{equation} \label{psio}
\begin{array}{rccccc}
\textrm{resp.}\quad \Psi &+& \beta \displaystyle \frac{\partial
\mathcal{D}^{(0)}_C[\Psi]}{\partial \nu} &=& - \beta \displaystyle
\frac{\partial H}{\partial \nu} \quad &\textrm{on} \, \partial C,
\end{array}
\end{equation}
where $\tilde{\nu}(\tilde{x})$ (resp. $\nu(x)$) denotes the
outward unit normal to $\partial C_{\epsilon}$ (resp. $\partial
C$) at $\tilde{x}$ (resp. $x$).

Therefore we obtain, for all $x \in \Omega$,
\begin{displaymath}
u_{\epsilon}(x) - u(x) =
\mathcal{D}^{(0)}_{\Omega}[u_{\epsilon}|_{\partial \Omega} -
u|_{\partial \Omega}](x) +
\mathcal{D}^{(0)}_{C_{\epsilon}}[\Psi_{\epsilon}](x) -
\mathcal{D}^{(0)}_C[\Psi](x),
\end{displaymath}
and, on $\partial \Omega$:
\begin{displaymath}
u_{\epsilon}(x) - u(x) = (\frac{I}{2} +
\mathcal{K}^{(0)}_{\Omega})[u_{\epsilon} - u](x) +
\mathcal{D}^{(0)}_{C_{\epsilon}}[\Psi_{\epsilon}](x) -
\mathcal{D}^{(0)}_C[\Psi](x).
\end{displaymath}

Our first step is to find high-order terms in the expansion of
$\Psi_{\epsilon}$. We define the operator $\mathcal{L}_{\epsilon}$
(resp. $\mathcal{L}$) on $L^2(\partial C_{\epsilon})$ (resp.
$L^2(\partial C)$) by
\begin{equation}
\begin{array}{ll}
\vspace{0.25 cm}&\mathcal{L}_{\epsilon}[f] = \displaystyle \frac{\partial
\mathcal{D}^{(0)}_{C_{\epsilon}}[f]}{\partial \tilde{\nu}}\\
\textrm{resp.} \qquad&\mathcal{L}[f] = \displaystyle
\frac{\partial \mathcal{D}^{(0)}_{C}[f]}{\partial \nu} .
\end{array}
\end{equation}

\begin{prop} \label{prop431}
Let $D$ be a bounded $\mathcal{C}^{2,\eta}$- domain in
$\mathbb{R}^2$, for $0<\eta<1$. We denote by $L_{D}$ the normal
derivative of the double layer potential on $D$, $L_D:= \partial
\mathcal{D}^{(0)}_D /  \partial \nu$. Then, $I + \beta \, L_D :
\mathcal{C}^{2,\eta} \to \mathcal{C}^{1,\eta}$ is a bounded
operator and has a bounded inverse.
\end{prop}

\begin{proof}
The boundness of $L_D : \mathcal{C}^{2,\eta} \to
\mathcal{C}^{1,\eta}$ is proved in \cite{14}. Note that since
$L_D$ is not a compact operator, we can not apply the Fredholm
alternative. However, $L_D$ is positive \cite{2} and the
proposition follows since $\beta >0$.
\end{proof}

For $f \in \mathcal{C}^{2,\eta}(\partial C_{\epsilon})$,
$\tilde{x} \in
\partial C_{\epsilon}$, $\mathcal{L}_{\epsilon}$ has the following
expression \cite{6}:
\begin{displaymath}
\begin{array}{ll}
\vspace{0.3cm} \displaystyle \frac{\partial
\mathcal{D}^{(0)}_{C_{\epsilon}}[f]}{\partial \nu}(x)
 =& - \displaystyle \frac{1}{2\pi}\int_{\partial D} \frac{\langle \tilde{\nu}(\tilde{x}),
 \tilde{\nu}(\tilde{y}) \rangle}{|\tilde{x} - \tilde{y}|^2} (f(\tilde{y})
 - f(\tilde{x}))\, ds_{\epsilon}(\tilde{y})\\
& + \displaystyle \frac{1}{\pi}\int_{\partial D} \frac{\langle
\tilde{x} - \tilde{y} , \tilde{\nu}(\tilde{x}) \rangle\langle
\tilde{x} - \tilde{y} , \tilde{\nu}(\tilde{y}) \rangle}{|\tilde{x}
- \tilde{y}|^4} (f(\tilde{y}) - f(\tilde{x}))
\,ds_{\epsilon}(\tilde{y}) .
\end{array}
\end{displaymath}

The outward unit normal to $\partial C$ at $x$, $\nu(x)$, and the
tangential vector, $T(x)$, are, in terms of polar coordinates:
 $$\nu(x) = {e_r}(x), \quad T(x) = {e_{\theta}(x)}.$$

The outward unit normal to $\partial C_{\epsilon}$ at $\tilde{x}$,
$\tilde{\nu}(\tilde{x})$, is given by
\begin{displaymath}
\tilde{\nu}(\tilde{x}) = \displaystyle
\frac{R_{-\frac{\pi}{2}}(\tilde{x'}(\theta))}{|\tilde{x'}(\theta)|},
\end{displaymath}
where $R_{-\frac{\pi}{2}}$ is rotation by $-\frac{\pi}{2}$. In our
case, we then have
\begin{equation}\label{eq:nu}
\tilde{\nu}(\tilde{x}) = \displaystyle \frac{(R + \epsilon
h(\theta)) {e_r} - \epsilon h'(\theta)  {e_{\theta}}}{\left((R+
\epsilon h(\theta))^2 + (\epsilon
h'(\theta))^2\right)^{\frac{1}{2}}} .
\end{equation}

We can expand $\tilde{\nu}(\tilde{x})$, for $x \in \partial C$, as
follows:

\begin{equation}\label{eq:decnu}
\tilde{\nu}(\tilde{x}) = \sum_{n=0}^{\infty}
\epsilon^n\nu^{(n)}(\theta),
\end{equation}
where the vector-valued functions $\nu^{(n)}$ are uniformly
bounded independently of $n$.

In particular, at first order, $\tilde{\nu}(\tilde{x})$, for
$\tilde{x} \in \partial C_{\epsilon}$, is given by

\begin{equation}\label{eq:decnu1}
 \tilde{\nu}(\tilde{x}) = {e_r} - \displaystyle \frac {h'(\theta)}{R} \,
 {e_{\theta}} + o(\epsilon).
\end{equation}

Set $\tilde{x}, \tilde{y} \in \partial C_{\epsilon}$. We have
\begin{equation}\label{eq:x-y}
\tilde{x} - \tilde{y} = R ({e_{r}({x})} - {e_{r}({y})}) + \epsilon
(h(\theta^x)\,{e_{r}({x})} - h(\theta^y)\,{e_{r}({y})}).
\end{equation}
If we denote
\begin{equation} \label{defc}
c = \cos (\theta^x - \theta ^y), \quad s = \sin (\theta^x - \theta
^y), \end{equation} then we obtain:
\begin{displaymath}
\begin{array}{ll}
\vspace{0.3 cm} |\tilde{x} - \tilde{y}|^2 = 2 R^2 (1 - c) +& 2 \epsilon R (1 - c) (h(\theta^x) + h(\theta^y)) \\
&\qquad \qquad  + \epsilon ^2\left(h(\theta^x)^2 + h(\theta^y)^2 -
2 h(\theta^x) h(\theta^y) c\right).
\end{array}
\end{displaymath}

and

\begin{equation}\label{eq:x-y2}
\displaystyle \frac{1}{|\tilde{x} - \tilde{y}|^2} =
\displaystyle\frac{1}{2 R^2 (1 - c)} \,\displaystyle \frac{1}{1 +
\epsilon F(\theta^x,\theta^y) + \epsilon^2 G(\theta^x,\theta^y)},
\end{equation}

where
\begin{displaymath}
F(\theta^x,\theta^y) =\displaystyle \frac {(h(\theta^x) +
h(\theta^y))}{R}, \quad G(\theta^x,\theta^y) = \displaystyle
\frac{\langle h(\theta^x)\,{e_{r}({x})} -
h(\theta^y)\,{e_{r}({y})} \rangle^2}{2 R^2 (1 - c)}.
\end{displaymath}

Likewise, we write
\begin{equation}\label{eq:x-y4}
\displaystyle \frac{1}{|\tilde{x} - \tilde{y}|^4} =
\displaystyle\frac{1}{4 R^4 (1 - c)^2} \,\displaystyle \frac{1}{(1
+ \epsilon F(\theta^x,\theta^y) + \epsilon^2
G(\theta^x,\theta^y))^2} .
\end{equation}

It follows, from (\ref{eq:nu}), (\ref{eq:sigma}) and
(\ref{eq:x-y2}), that
\begin{displaymath}
\begin{array}{l}
\vspace{0.3cm}\displaystyle  \frac{\langle \tilde{\nu}(\tilde{x}), \tilde{\nu}(\tilde{y}) \rangle}{|\tilde{x} - \tilde{y}|^2}\,
 ds_{\epsilon}(\tilde{y}) = \displaystyle \frac{K_0 + \epsilon K_1 + \epsilon^2 K_2}{2 R^2 (1 - c)}\\
\qquad \qquad  \times  \displaystyle \frac{1}{1 + \epsilon
F(\theta^x,\theta^y) + \epsilon^2 G(\theta^x,\theta^y)}
\frac{R}{\left((R+ \epsilon h(\theta^x))^2 + (\epsilon
h'(\theta^x))^2\right)^{\frac{1}{2}}} R d\theta^y,
\end{array}
\end{displaymath}

where
\begin{displaymath}
\begin{array}{l}
\vspace{0.5 cm} K_0 = c,\\
\vspace{0.5 cm} K_1 = \displaystyle \frac{1}{R} \left [ (h(\theta^x)
+ h(\theta^y))c + (h'(\theta^x) - h'(\theta^y))s \right], \\
K_2 =  \displaystyle \frac{h'(\theta^x) h'(\theta^y)}{R^2} c.
\end{array}
\end{displaymath}

One can see, from the previous formulas, that the singularity of
$\displaystyle \frac{K_i}{2 R^2 (1 - c)}$ for $i \in [0,2]$ is of
order $O(|\theta^x-\theta^y|^{-2})$, since $1 - c =
O(|\theta^x-\theta^y|^{-2})$.

Likewise, thanks to (\ref{eq:nu}), (\ref{eq:sigma}) and
(\ref{eq:x-y2}), we can explicit $M_i$ for $i \in [0,4]$ such that
\begin{displaymath}
\begin{array}{l}
\vspace{0.3cm}\displaystyle \frac{\langle \tilde{x} - \tilde{y} ,
\tilde{\nu}(\tilde{x}) \rangle\langle  \tilde{x} - \tilde{y} ,
\tilde{\nu}(\tilde{y}) \rangle}{|\tilde{x} - \tilde{y}|^4} \,
ds_{\epsilon}(\tilde{y}) = \displaystyle \frac{M_0 + \epsilon M_1
+ \epsilon^2 M_2 +
\epsilon^3 M_3 + \epsilon^4 M_4 }{4 R^4 (1 - c)^2}\\
\qquad \qquad \times  \displaystyle \frac{1}{(1 + \epsilon
F(\theta^x,\theta^y) + \epsilon^2 G(\theta^x,\theta^y))^2}
\frac{R}{\left((R+ \epsilon h(\theta^x))^2 + (\epsilon
h'(\theta^x))^2\right)^{\frac{1}{2}}} R d\theta^y,
\end{array}
\end{displaymath}
and the singularity of $\displaystyle
 \frac{M_i}{4 R^4 (1 - c)^2}$ for $i \in [0,4]$ is of order $O(|\theta^x-\theta^y|^{-2})$.
Therefore, we get
\begin{displaymath}
\begin{array}{l}
\vspace{0.3cm} L_{\epsilon} \,ds_{\epsilon}(\tilde{y}) = \displaystyle \frac{N_0 + \epsilon N_1 + \epsilon^2 N_2 + \epsilon^3 N_3 + \epsilon^4 N_4}{2 R^4 (1 - c)^2}\\
\qquad \quad  \times \displaystyle \frac{1}{(1 + \epsilon
F(\theta^x,\theta^y) + \epsilon^2 G(\theta^x,\theta^y))^2}
\frac{R}{\left((R+ \epsilon h(\theta^x))^2 + (\epsilon
h'(\theta^x))^2\right)^{\frac{1}{2}}} \,R d\theta^y,
\end{array}
\end{displaymath}
where $L_{\epsilon} : = - \displaystyle  \frac{\langle
\tilde{\nu}(\tilde{x}), \tilde{\nu}(\tilde{y}) \rangle}{|\tilde{x}
- \tilde{y}|^2} + 2 \displaystyle \frac{\langle \tilde{x} -
\tilde{y} , \tilde{\nu}(\tilde{x}) \rangle\langle  \tilde{x} -
\tilde{y} , \tilde{\nu}(\tilde{y}) \rangle}{|\tilde{x} -
\tilde{y}|^4}$ is the kernel of $\mathcal{L}_{\epsilon}$ and the
singularity of $\displaystyle
 \frac{N_i}{2 R^4 (1 - c)^2}$ for $i \in [0,4]$ is of order $O(|\theta^x-\theta^y|^{-2})$.
We  do not give here the expressions of $N_2,N_3,N_4$ due to their
length, but $N_0$ and $N_1$ are given by
\begin{displaymath}
\begin{array}{rl}
\vspace{0.5 cm} N_0 &=  - R^2 (1 - c),\\
\vspace{0.5 cm} N_1 &= - 2  R (1-c) (h(\theta^x) + h(\theta^y)).
\end{array}
\end{displaymath}
Recall that
\begin{displaymath}
F(\theta^x,\theta^y) =\displaystyle \frac {(h(\theta^x) +
h(\theta^y))}{R}, \quad G(\theta^x,\theta^y) =  \displaystyle
\frac{(h(\theta^x) - h(\theta^y)^2 + 2 h(\theta^x)
h(\theta^y)(1-c)}{2 R^2 (1 - c)} .
\end{displaymath}
We introduce the following series, which converges absolutely and
uniformly,
\begin{displaymath}
\displaystyle \frac{1}{(1 + \epsilon F(\theta^x,\theta^y) +
\epsilon^2 G(\theta^x,\theta^y))^2} \frac{R}{\left((R+ \epsilon
h(\theta^x))^2 + (\epsilon h'(\theta^x))^2\right)^{\frac{1}{2}}} =
\sum_{p=0}^{\infty} \epsilon^p F_p(\theta^x, \theta^y) .
\end{displaymath}
The first order term is given by
\begin{equation}\label{eq:F1}
F_1(\theta^x, \theta^y) = - \displaystyle \frac {(3 h(\theta^x) +
2 h(\theta^y))}{R} .
\end{equation}
Note that $(F_p)_{p \in \mathbb{N}}$, like $F$ and $G$, have no
singularity and are uniformly bounded.

We define the following functions, for all $x,y \in \partial C$:
\begin{displaymath}
\begin{array}{ll}
\vspace{0.4 cm} L^{(0)} =  \displaystyle \frac{N_0 }{2 R^4 (1 -
c)^2}, & L^{(1)} =
\displaystyle \frac{N_0 F_1 + N_1}{2 R^4 (1 - c)^2} ,\\
L^{(2)} =  \displaystyle \frac{N_0 F_2 + N_1 F_1 + N_2}{2 R^4 (1 -
c)^2} , \qquad \quad  & L^{(3)} =  \displaystyle \frac{N_0 F_3 +
N_1 F_2 + N_2 F_1 + N_3}{2 R^4 (1 - c)^2},
\end{array}
\end{displaymath}
 and, for $n \geq 4$,
\begin{equation}\label{eq:ln}
L^{(n)} = \displaystyle \frac{1}{2 R^4 (1 - c)^2} \left( N_0 F_n +
N_1 F_{n-1} + N_2 F_{n-2} + N_3 F_{n-3} + N_4 F_{n-4}\right).
\end{equation}
Thanks to the explicit formulas of $(N_{i})_{i \in [0,4]}$ and
(\ref{eq:F1}), we obtain in particular that, for all $x,y \in
\partial C$,
\begin{equation}\label{eq:L01}
L^{(0)} = - \displaystyle \frac{1}{2 R^3 (1 - c)} \quad
\textrm{and} \quad L^{(1)} = \displaystyle \frac{h(\theta^x)}{2
R^3 (1 - c)} ,
\end{equation}
where $c$ is given by (\ref{defc}).

By construction, $L^{(n)}$, for all $n \in \mathbb{N}$, have a
singularity of order $O(|\theta^x-\theta^y|^{-2})$.

The integral operators $(\mathcal{L}^{(n)})_{n \in \mathbb{N}}$,
associated to the kernels $(L^{(n)})_{n \in \mathbb{N}}$, are
given, for all $f \in \mathcal{C}^{2,\eta}(\partial C)$, $x \in
\partial C$, by
\begin{displaymath}
\mathcal{L}^{(n)}[f](x) = \displaystyle \frac{1}{2\pi}
\int_0^{2\pi} L^{(n)}(\theta^x,\theta^y) (f(\theta^y) -
f(\theta^x)) R d\theta^y.
\end{displaymath}

It follows from (\ref{eq:L01}) that, for all
$\mathcal{C}^{2,\eta}(\partial C)$, $x \in \partial C$:
\begin{equation}\label{eq:opL01}
\mathcal{L}^{(0)}[f](x) = \mathcal{L}[f](x) \quad \textrm{and}
\quad \mathcal{L}^{(1)}[f](x) = - h(\theta^x)  \mathcal{L}[f](x).
\end{equation}

We can now write, from our construction, an expansion of
$\mathcal{L}_{\epsilon}$.
\begin{prop} \label{prop432}
Let $N \in \mathbb{N}$. There exists $C$ depending only on R and
$||h||_{\mathcal{C}^2}$\, , such that, for any $\tilde{f} \in
\mathcal{C}^{2,\eta}(\partial C_{\epsilon})$, $0<\eta<1$, we have
\begin{displaymath}
\big || \mathcal{L}_{\epsilon}[\tilde{f}]\circ \tau_{\epsilon} -
\mathcal{L}[f] - \sum_{n=0}^N \epsilon^n \mathcal{L}^{(n)}[f] \big
||_{\mathcal{C}^{1,\eta}(\partial C)} \leq C \epsilon^{N+1}
||f||_{\mathcal{C}^{2,\eta}(\partial C)},
\end{displaymath}
where $\tau_{\epsilon}$ is the diffeomorphism from $\partial C$
onto $\partial C_{\epsilon}$ given by $\tau_{\epsilon}(x) =
\tilde{x}$ and the function $f$ is defined by $f:=\tilde{f} \circ
\tau_{\epsilon}$.
\end{prop}

\begin{proof}
Let $f \in \mathcal{C}^{2,\eta}$. We know that $\displaystyle
\frac{N_i}{2 R^4 (1 - c)^2}$, for all $i \in [0, 4]$, have a
singularity of order $O(|\theta^x-\theta^y|^{-2})$.

Thanks to the $\mathcal{C}^1$-character of $f$,
$(\theta^x,\theta^y) \to \displaystyle \frac{N_i}{2 R^4 (1 -
c)^2}(f(\theta^y) - f(\theta^x))$ have a singularity of order
$O(|\theta^x-\theta^y|^{-1})$.

Besides the Hilbert transform is a bounded operator from
$\mathcal{C}^{0,\eta}$ to $ \mathcal{C}^{0,\eta}$. From the
boundness of $h$ and its derivatives, it follows that the
operators associated with the kernels $\displaystyle \frac{N_i}{2
R^4 (1 - c)^2}$ for $i \in [0, 4]$ are bounded from
$\mathcal{C}^{2,\eta}$ to $\mathcal{C}^{1,\eta}$.

Since the $(F_p)_{p \in \mathbb{N}}$ are uniformly bounded, the
construction of $L^{(n)}$ (\ref{eq:ln}) implies that there exists
a constant $K(R, ||h||_{\mathcal{C}^2})$  such that
\begin{displaymath}
|| \mathcal{L}^{(n)}[f]||_{\mathcal{C}^{0,\eta}(\partial C)}  \leq
K ||f'||_{\mathcal{C}^{0, \eta}(\partial C)},
\end{displaymath}
where $f^\prime$ is the derivative of $f$ with respect to
$\theta$.  Likewise, since the kernel of
$\mathcal{L}^{(n)}[f]^\prime (x)$ is of order $O\left (
\displaystyle \frac{f(y) - f(x) - (x-y) f'(x)}{|x-y|^2} \right )$,
the $\mathcal{C}^2$-character of $f$ gives us a singularity of
order $O(|\theta^x-\theta^y|^{-1})$. We therefore obtain that
\begin{displaymath}
|| \mathcal{L}^{(n)}[f]^\prime ||_{\mathcal{C}^{0,\eta}(\partial
C)} \leq \tilde{K} ||f''||_{\mathcal{C}^{0, \eta}(\partial C)},
\end{displaymath}
where $\tilde{K}(R, ||h||_{\mathcal{C}^2})$ is a constant and
$f^{''}$ is the second derivative of $f$. Therefore, there exists
a constant $\widehat{K}(R, ||h||_{\mathcal{C}^2})$ such that
\begin{displaymath}
|| \mathcal{L}^{(n)}[f]||_{\mathcal{C}^{1,\eta}(\partial C)} \leq
\widehat{K} ||f||_{\mathcal{C}^{2, \eta}(\partial C)} .
\end{displaymath}
For all $n \in \mathbb{N}$, the operator $\mathcal{L}^{(n)} :
\mathcal{C}^{2,\eta} \to \mathcal{C}^{1,\eta} $ is bounded and the
constant $\widehat{K}$ does not depend on $n$. Let $N \in
\mathbb{N}$. Let $\tilde{f} \in \mathcal{C}^{2,\eta}(\partial
C_{\epsilon})$. We introduce $f:=\tilde{f} \circ \tau_{\epsilon}$,
$f \in \mathcal{C}^{2,\eta}(\partial C)$. We have
\begin{displaymath}
||\sum_{n=N+1}^{\infty} \epsilon^n \mathcal{L}^{(n)}[f]
||_{\mathcal{C}^{1,\eta}(\partial C)} \leq \displaystyle \frac{
\epsilon^{N+1}}{1 - \epsilon}\,  \widehat{K}  \,
||f||_{\mathcal{C}^{2,\eta}(\partial C)},
\end{displaymath}
which ends the proof of the result.
\end{proof}

By substituting the result of Proposition \ref{prop432} into the
integral equation (\ref{psieps}) verified by $\Psi_{\epsilon}$, we
obtain for all $N \in \mathbb{N}$ that
\begin{equation}\label{eq:psidec}
\forall x \!\!\in \!\partial C , \,\,\, (I + \beta \mathcal{L} +
\beta \sum_{n=0}^N \epsilon^n
\mathcal{L}^{(n)})[\Psi_{\epsilon}](\tilde{x}) + o(\epsilon^N) = -
\beta \displaystyle \frac{\partial H_{\epsilon}}{\partial
\tilde{\nu}}(\tilde{x}).
\end{equation}

We use Taylor-Lagrange's theorem and (\ref{eq:decnu}) to expand
$\displaystyle \frac{\partial H_{\epsilon}}{\partial
\tilde{\nu}}(\tilde{x})$:
\begin{equation}\label{eq:deche}
\displaystyle \frac{\partial H_{\epsilon}}{\partial
\tilde{\nu}}(\tilde{x}) = \displaystyle \left (\sum_{p=0}^{\infty}
\sum_{|\alpha| = p} \frac{\epsilon^p}{\alpha!} (\partial^{\alpha}
\nabla H_{\epsilon})(x) (h(\theta)\nu(x))^{\alpha}\right ) \left
(\sum_{p=0}^{\infty} \epsilon^p\nu^{(p)}(\theta) \right ).
\end{equation}

In particular, at first order, we have
\begin{equation}\label{eq:he1}
\displaystyle \frac{\partial H_{\epsilon}}{\partial
\tilde{\nu}}(\tilde{x}) = \frac{\partial H_{\epsilon}}{\partial
r}(x) + \epsilon \left( - \displaystyle
\frac{h'(\theta)}{R^2}\frac{\partial H_{\epsilon}}{\partial
\theta}(x) + h(\theta) \frac{\partial^2 H_{\epsilon}}{\partial
r^2}(x) \right).
\end{equation}

Our integral equation (\ref{eq:psidec}) then becomes
\begin{equation}\label{eq:psidec2}
\forall x \!\!\in\! \partial C, \,\,\,(I + \beta \mathcal{L} +
\beta \sum_{p=0}^N \epsilon^n
\mathcal{L}^{(n)})[\Psi_{\epsilon}](\tilde{x}) + o(\epsilon^N) = -
\beta \sum_{n=0}^{\infty} \epsilon^n G_n(x),
\end{equation}

where $(G_n)_{n \in \mathbb N}$ are the coefficients in the
expansion (\ref{eq:deche}).

Equation (\ref{eq:psidec2}) can therefore be solved recursively in
the following way:
\begin{equation}\label{eq:psif}
\begin{array}{cl}
\vspace{0.2 cm}&\Psi^{(0)} = -\beta (I + \beta \mathcal{L})^{-1} \left[G_0 \right],\\
\forall n \leq N,  \quad & \Psi^{(n)} = - \beta (I + \beta
\mathcal{L})^{-1} \left[G_n + \displaystyle \sum_{p=0}^{n-1}
\mathcal{L}^{(n-p)}\Psi^{(p)} \right].
\end{array}
\end{equation}

In particular, thanks to (\ref{eq:opL01}) and (\ref{eq:he1}), we
have
\begin{equation}
\begin{array}{l}
\vspace{0.3 cm} \Psi^{(0)} =  - \beta (I + \beta \mathcal{L})^{-1} \left( \displaystyle \frac{\partial H_{\epsilon}}{\partial \nu}\right), \\
\Psi^{(1)} = - \beta (I + \beta \mathcal{L})^{-1} \left( -
\displaystyle \frac{h'}{R^2}\frac{\partial H_{\epsilon}}{\partial
\theta} + h \frac{\partial^2 H_{\epsilon}}{\partial r^2} - h
\displaystyle \frac{\partial}{\partial
\nu}\mathcal{D}^{(0)}_{C}[\Psi^{(0)}]\right).
\end{array}
\end{equation}

We obtain the following proposition.

\begin{prop} \label{prop433}
Let $N \in \mathbb N$. There exists $K$, depending only on $N$,
$R$ and the $\mathcal{C}^2$- norm of $h$, such that
\begin{equation}\label{eq:decpsi}
|| \Psi_{\epsilon} - \sum_{n=0}^N \epsilon^n\,
\Psi^{(n)}||_{\mathcal{C}^{2, \eta}(\partial C)} \leq K \epsilon
^{N+1},
\end{equation}
where $(\Psi^{(n)})_{n\leq N}$ are defined by the recursive
relation (\ref{eq:psif}).
\end{prop}

In order to prove Proposition \ref{prop433}, we need the
 following result \cite[Theorem 1.16]{15}.
\begin{lem} \label{lemspec} Let $X$ and $Y$ be two Banach spaces.
Let $T$ and $A$ be two operators from $X$ to $Y$, such that $D(T)
\subset D(A)$, where $D(T)$ and $D(A)$ are the domains of $T$ and
$A$, respectively. Let $T^{-1}$ exist and be a bounded operator
from $Y$ to $X$ (so that $T$ is closed). We suppose that two
positive constants $a, b$ exist such that
\begin{displaymath}
\begin{array}{rl}
\vspace{0.3cm} \forall u \in D(T),  &||Au|| \leq a ||u|| + b ||Tu||, \\
&a ||T^{-1}|| + b < 1.
\end{array}
\end{displaymath}
Then $S = T + A$ is closed and invertible, $S^{-1}$ is a bounded
operator from $Y$ to $X$ and the following inequalities hold:
\begin{displaymath}
||S^{-1}|| \leq \displaystyle \frac{||T^{-1}||}{1 - a ||T^{-1}|| -
b}, \quad ||S^{-1} - T^{-1}|| \leq \displaystyle
\frac{||T^{-1}||(a ||T^{-1}|| + b)}{1 - a ||T^{-1}|| - b}.
\end{displaymath}
If in addition $T^{-1}$ is compact, so is $S^{-1}$.
\end{lem}

\begin{proof}[Proof of Proposition \ref{prop433}]
By definition, $\Psi_{\epsilon}$ verifies:
$$(I + \beta \mathcal{L}_{\epsilon}) [\Psi_{\epsilon}] = - \beta \sum_{n=0}^{\infty} \epsilon^n G_n .$$
Besides, it follows, from our recursive construction of the
$(\Psi^{(i)})_{i \in [0,N]}$, that
\begin{displaymath}
(I + \beta \mathcal{L} + \beta \sum_{n=1}^N \epsilon^n
\mathcal{L}^{(n)})[\sum_{n=0}^N \epsilon^p \Psi^{(p)}] =  - \beta
\sum_{n=0}^{\infty} \epsilon^n G_n + \epsilon^{N+1} A_N,
\end{displaymath}
\vspace{0.2cm}where $A_N = \displaystyle \sum_{n=0}^{N}
\epsilon^n\, \sum_{p=0}^{N+n} \mathcal{L}^{(N+1+n-p)}\,
[\Psi^{(p)}] + \beta \sum_{n=0}^{\infty} \epsilon^n \, G_{N+1+n}.$

Therefore, we have
\begin{equation}\label{eq:diffpsi}
\begin{array}{ll}
 \Psi_{\epsilon} - \displaystyle \sum_{n=0}^N \epsilon^n\, \Psi^{(n)} = &\left ( (I +
 \beta \mathcal{L}_{\epsilon})^{-1} -  (I +  \beta \mathcal{L} + \beta \displaystyle
 \sum_{n=1}^N \epsilon^n \mathcal{L}^{(n)})^{-1} \right ) [ - \beta \displaystyle \sum_{n=0}^{\infty} \epsilon^n G_n]  \\
& - (I +  \beta \mathcal{L} + \beta \displaystyle \sum_{n=1}^N
\epsilon^n \mathcal{L}^{(n)})^{-1} [\epsilon^{N+1}   A_N].
 \end{array}
\end{equation}

We know from Proposition \ref{prop431} that the bounded operator
$T := I + \beta \mathcal{L}_{\epsilon} : \mathcal{C}^{2, \eta} \to
\mathcal{C}^{1, \eta}$ has a bounded inverse $T^{-1}:
\mathcal{C}^{1, \eta} \to \mathcal{C}^{2, \eta}$. We define
\begin{displaymath}
A:= \beta \mathcal{L} + \beta \sum_{n=1}^N \epsilon^n
\mathcal{L}^{(n)} - \beta \mathcal{L}_{\epsilon}.
\end{displaymath}
From Proposition \ref{prop432}, it follows that there exists a
constant $C(R, ||h||_{\mathcal{C}^2})$ such that
\begin{displaymath}
||A[u]||_{\mathcal{C}^{1, \eta}(\partial C)} \leq C \epsilon^{N+1}
||u||_{\mathcal{C}^{2, \eta}(\partial C)}.
\end{displaymath}
For $\epsilon$ small enough, we have
$$ C \epsilon^{N+1} ||T^{-1}|| <1.$$
In the following, we apply Lemma \ref{lemspec} with $a:=C \,
\epsilon^{N+1}$ and $b:=0$.

The operator $S:= I +  \beta \mathcal{L} + \beta \displaystyle
\sum_{n=1}^N \epsilon^n \mathcal{L}^{(n)}$ has a bounded inverse,
which satisfies:
\begin{displaymath}
\begin{array}{ll}
\vspace{0.2 cm} &||(I +  \beta \mathcal{L} + \beta \displaystyle
\sum_{n=1}^N \epsilon^n \mathcal{L}^{(n)})^{-1}|| \leq \displaystyle
\frac{||T^{-1}||}{1 - C \epsilon^{N+1} ||T^{-1}||},\\
\textrm{and} \qquad &||(I +  \beta \mathcal{L} + \beta
\displaystyle  \sum_{n=1}^N \epsilon^n \mathcal{L}^{(n)})^{-1} -
(I + \beta \mathcal{L}_{D})^{-1}|| \leq \displaystyle \frac{C
\epsilon^{N+1} ||T^{-1}||^2}{1 - C \epsilon^{N+1}||T^{-1}||}.
\end{array}
\end{displaymath}
We use (\ref{eq:diffpsi}) to get
\begin{displaymath}
\bigg|\bigg| \Psi_{\epsilon} - \sum_{n=0}^N \epsilon^n\,
\Psi^{(n)}\bigg|\bigg|_{\mathcal{C}^{2, \eta}}\leq  \displaystyle
\frac{\epsilon^{N+1} ||T^{-1}||}{1 - C \epsilon^{N+1} ||T^{-1}||}
\left(C ||T^{-1}|| \bigg|\bigg| \beta \displaystyle \frac{\partial
H_{\epsilon}}{\partial \tilde{\nu}}\bigg|\bigg|_{\mathcal{C}^{1,
\eta}}+ ||A_N||_{\mathcal{C}^{1, \eta}}\right).
\end{displaymath}
Recall that $H_{\epsilon}$ is $\mathcal{C}^{\infty}$ on $\partial
C$. Hence, for all $p \in \mathbb{N}$, $G_p$ is bounded. From
Proposition \ref{prop432}, we know that $\mathcal{L}^{(n)} :
\mathcal{C}^{2, \eta}(\partial C) \to \mathcal{C}^{1,
\eta}(\partial C)$, for all $n \in \mathbb{N}$, are bounded
operators. We have also, from Proposition \ref{prop431}, that $(I
+ \beta \mathcal{L})^{-1} : \mathcal{C}^{1, \eta}(\partial C) \to
\mathcal{C}^{2, \eta}(\partial C)$ is bounded. One can prove
recursively, from the construction (\ref{eq:decpsi}), that, for
all $p \in \mathbb{N}$, $\Psi^{(p)}$ is $\mathcal{C}^{2,
\eta}(\partial C)$ - bounded. $A_N$ and $ \displaystyle
\frac{\partial H_{\epsilon}}{\partial \tilde{\nu}}$ are therefore
$\mathcal{C}^{1, \eta}(\partial C)$ - bounded.

Finally, we obtain that there exists a constant $K(N, R,
||h||_{\mathcal{C}^2})$ such that
\begin{displaymath}
\bigg|\bigg| \Psi_{\epsilon} - \sum_{n=0}^N \epsilon^n\,
\Psi^{(n)}\bigg|\bigg|_{\mathcal{C}^{2, \eta}} \!\!\leq  K \,
\epsilon^{N+1},
\end{displaymath}
and the proof of Proposition \ref{prop433} is complete.
\end{proof}

We now explicit the first order term in the expansion of
$\widetilde{c_{\textrm{flr}}}$ as function of the cell membrane
perturbation.  For doing so, we introduce, for $n\in \mathbb{N}
\setminus \{0\}$ and $x \in
\partial \Omega$:
\begin{equation}\label{eq:vn}
v_n(x) := \displaystyle \sum_{i+j+k+l=n} \int_0^{2\pi}
\frac{h(y)^{i}}{i!}\,\left(\nabla_y \big(
\frac{\partial^{i}}{\partial r_y^i}\Gamma^{(0)}(x,y)\big) \cdot
\nu^{(j)}(y) \right)\Psi^{(k)}(\theta^y) \sigma^{(l)}(\theta^y)
d\theta^y.
\end{equation}
It follows from (\ref{eq:sigma1}), (\ref{eq:decnu1}),
(\ref{eq:decpsi}) and (\ref{eq:vn}), that for all $x\in \partial
\Omega$:
\begin{displaymath}
\begin{array}{rcl}
\vspace{0.3 cm} v_1(x) =&\displaystyle  \int_0^{2\pi}
\displaystyle \frac{\partial^2}{\partial r_y^2}\Gamma^{(0)}(x,y)
h(\theta^y) \Psi^{(0)}(\theta^y) R d\theta^y &- \displaystyle
\frac{1}{R} \displaystyle \int_0^{2\pi}  \displaystyle
\frac{\partial}{\partial \theta^y}\Gamma^{(0)}(x,y)
\Psi^{(0)}(\theta^y) h'(\theta^y) d\theta^y
 \\
+& \displaystyle \int_0^{2\pi} \displaystyle
\frac{\partial}{\partial r_y}\Gamma^{(0)}(x,y)
\Psi^{(1)}(\theta^y) R d\theta^y &+ \displaystyle \int_0^{2\pi}
\displaystyle \frac{\partial}{\partial r_y}\Gamma^{(0)}(x,y)
\Psi^{(0)}(\theta^y) h(\theta^y) d\theta^y .
\end{array}
\end{displaymath}
In terms of polar coordinates, the Laplacian has the following
expression:
\begin{displaymath}
\Delta = \displaystyle \frac{\partial^2}{\partial r^2} +
\frac{1}{r} \frac{\partial}{\partial r} + \frac{1}{r^2}
\frac{\partial^2}{\partial \theta^2}.
\end{displaymath}

Therefore, we have for all $x \in \partial \Omega$:
\begin{displaymath}
\begin{array}{rl}
\vspace{0.3 cm} v_1(x) =& - \displaystyle  \frac{1}{R}
\int_0^{2\pi} \displaystyle \frac{\partial^2}{\partial
{\theta^y}^2}\Gamma^{(0)}(x,y) h(\theta^y) \Psi^{(0)}(\theta^y)
d\theta^y -  \displaystyle  \frac{1}{R}\displaystyle \int_0^{2\pi}
\displaystyle \frac{\partial}{\partial \theta^y}\Gamma^{(0)}(x,y)
\Psi^{(0)}(\theta^y) h'(\theta^y) d\theta^y
 \\
+& \displaystyle \int_0^{2\pi} \displaystyle
\frac{\partial}{\partial r_y}\Gamma^{(0)}(x,y)
\Psi^{(1)}(\theta^y) R d\theta^y .
\end{array}
\end{displaymath}
Besides, we obtain, thanks to (\ref{eq:decpsi}) and (\ref{eq:vn}),
that
\begin{displaymath}
\mathcal{D}^{(0)}_{C_{\epsilon}}[\Psi_{\epsilon}](x) = - \beta
\,\mathcal{D}^{(0)}_C(I + \beta \mathcal{L})^{-1}\left[
\displaystyle \frac{\partial H_{\epsilon}}{\partial \nu}
\bigg|_{\partial C}\right] + \sum_{n=1}^N \epsilon^n v_n(x)+
o(\epsilon^N) .
\end{displaymath}

The integral equation (\ref{psio}) that $\Psi$ verifies, then
gives us
\begin{displaymath}
\mathcal{D}^{(0)}_{C_{\epsilon}}[\Psi_{\epsilon}] -
\mathcal{D}^{(0)}_{C} [\Psi] = - \beta \,\mathcal{D}^{(0)}_C(I +
\beta \mathcal{L})^{-1}\left[ \displaystyle \frac{\partial
H_{\epsilon}}{\partial \nu} \bigg|_{\partial C} -  \displaystyle
\frac{\partial H}{\partial \nu} \bigg|_{\partial C}\right] +
\sum_{n=1}^N \epsilon^n v_n + o(\epsilon^N).
\end{displaymath}

By definition, we have on $\partial C$
\begin{displaymath}
H_{\epsilon} - H =
\mathcal{D}^{(0)}_{\Omega}[u_{\epsilon}|_{\partial \Omega} -
u|_{\partial \Omega}].
\end{displaymath}

Let $\mathcal{E}$ be the operator defined by
\begin{equation} \label{defE}
\mathcal{E}[v](x) := \beta \,\mathcal{D}^{(0)}_C(I + \beta
\mathcal{L})^{-1} \left[ \displaystyle \frac{\partial}{\partial
\nu}(\mathcal{D}^{(0)}_{\Omega}v)\bigg|_{\partial C} \right](x) -
(\frac{I}{2} + \mathcal{K}^{(0)}_{\Omega})[v](x),
\end{equation}
for all $v \in L_0^2(\partial \Omega)$ and $x \in \partial
\Omega$.

Recall that on $\partial \Omega$:
\begin{displaymath}
u_{\epsilon}(x) - u(x) = (\frac{I}{2} +
\mathcal{K}^{(0)}_{\Omega})[u_{\epsilon} - u](x) +
\mathcal{D}^{(0)}_{C_{\epsilon}}[\Psi_{\epsilon}](x) -
\mathcal{D}^{(0)}_C[\Psi](x).
\end{displaymath}
We obtain, for all $x \in \partial \Omega$,  that
\begin{equation} \label{defE2}
(I + \mathcal{E})[u_{\epsilon} - u](x) =  \sum_{n=1}^N \epsilon^n
v_n(x)+ o(\epsilon^N),
\end{equation}
and, at first order,
\begin{displaymath}
(I + \mathcal{E})[u_{\epsilon} - u](x) =  \epsilon \, v_1(x) +
o(\epsilon),
\end{displaymath}
where $v_1$ is given by the formula:
\begin{equation}\label{eq:v1}
 v_1(x) = - \displaystyle  \frac{1}{R} \int_0^{2\pi} \displaystyle \frac{\partial}{\partial {\theta^y}}
 \left(h(\theta^y) \displaystyle \frac{\partial}{\partial \theta^y} \Gamma^{(0)}(x,y)\right)
 \Psi^{(0)}(\theta^y) d\theta^y + \mathcal{D}^{(0)}_{C}[\Psi^{(1)}](x).
 \end{equation}

\begin{prop} \label{prop434}
Let $\mathcal{E}$ be defined by (\ref{defE}). The operator $I +
\mathcal{E}$ is invertible on $L^2_0(\partial \Omega)$.
\end{prop}

\begin{proof}
The operator $\mathcal{E}$ is compact. We can therefore apply the
Fredholm alternative. Let us prove the injectivity of $I +
\mathcal{E}$. For doing so, we introduce the function $v$ defined
on $\Omega$ by
\begin{displaymath}
v(x) = \mathcal{D}^{(0)}_{\Omega}[v|_{\partial \Omega}] -  \beta
\,\mathcal{D}^{(0)}_C(I + \beta \mathcal{L})^{-1} \left[
\displaystyle \frac{\partial}{\partial
\nu}(\mathcal{D}^{(0)}_{\Omega}[v])\bigg|_{\partial C} \right].
\end{displaymath}
It follows from Proposition \ref{prop421} that $v$ is solution to
(\ref{eq:u2}) with $H = \mathcal{D}^{(0)}_{\Omega}[v|_{\partial
\Omega}]$. The decomposition of the representation formula of such
a solution is unique so that we have
$\mathcal{S}^{(0)}_{\Omega}[\displaystyle \frac{\partial
v}{\partial \nu} |_{\partial \Omega}] = 0$ and hence
$\displaystyle \frac{\partial v}{\partial \nu} \bigg |_{\partial
\Omega} = 0$. Since $v$ is harmonic, we obtain that $v$ is
constant in $\Omega$. Recall that $\displaystyle \int_{\partial
\Omega} v = 0$. Therefore, we have $v = 0$ in $\Omega$. Besides,
on $\partial \Omega$, $v$ verifies:
\begin{displaymath}
\forall x \in \partial \Omega,  \qquad v(x) = - \mathcal{E}[v](x).
\end{displaymath}
We have proved the injectivity and hence invertibility of $I +
\mathcal{E}$ on $L^2_0(\partial \Omega)$.
\end{proof}

Now, combining Proposition \ref{prop434} and (\ref{defE2}) yields
\begin{displaymath}
u_{\epsilon}(x) - u(x) =  \sum_{n=1}^N \epsilon^n (I +
\mathcal{E})^{-1}[v_n](x)+ o(\epsilon^N).
\end{displaymath}

Note that by construction $\Psi^{(n)}$ and so $v_n$ still depend
on $\epsilon$. We can remove this dependance from our asymptotic
formula in the following way. We introduce $(G_n^0)_{n \in
\mathbb{N}}$ the expansion of $\displaystyle \frac{\partial
H}{\partial \tilde{\nu}}$. Let $(v_n^0)_{n \in \mathbb{N}
\setminus \{0\}}$ and $(\Psi^{(n)}_0)_{n \in \mathbb{N}}$ be
defined by (\ref{eq:vn}) and (\ref{eq:psif}), where $(G_n)_{n \in
\mathbb{N}}$ is replaced respectively by $(G_n^0)_{n \in
\mathbb{N}}$. We then obtain that
\begin{displaymath}
\begin{array}{ll}
\vspace{0.3 cm} \forall x \in \partial C,\quad & \Psi_{\epsilon}(x) = \Psi^{(0)}_0(x) + o(1), \\
\forall x \in \partial \Omega,\quad & u_{\epsilon}(x) = u(x) +
o(1).
\end{array}
\end{displaymath}
By repeating the same procedure with $H + \epsilon
\,\mathcal{D}^{(0)}_{\Omega}(I + \mathcal{E})^{-1}[v^0_1]$ instead
of $H$, one finds $(v^1_n)_{n \in \mathbb N^*}$ and
$(\Psi^{(n)}_1)_{n \in \mathbb{N}}$ such that
\begin{displaymath}
\begin{array}{ll}
\vspace{0.3 cm} \forall x \in \partial C,\quad & \Psi_{\epsilon}(x) = \Psi^{(0)}_1(x) + \epsilon \,\Psi^{(1)}_1(x) + o(\epsilon), \\
\forall x \in \partial \Omega,\quad & u_{\epsilon}(x) = u(x) +
\epsilon (I + \mathcal{E})^{-1}[v^1_1] +o(\epsilon).
\end{array}
\end{displaymath}

One can prove the following proposition, by repeating the same
procedure until one obtains $(v_n^{N})_{n \in \mathbb{N} \setminus
\{0\}}$.
\begin{prop}
Let $(v^N_n)_{n \in [1, N]}$ and $(\Psi^{(n)}_N)_{n \in [0,N]}$ be
\vspace{0.15 cm}  the functions defined above. The following
asymptotic formulas hold:
\begin{displaymath}
\begin{array}{ll}
\forall x \in \partial C,\quad & \Psi_{\epsilon}(x) =  \displaystyle \sum_{n=1}^N \epsilon^n \, \Psi^{(n)}_N+ o(\epsilon^N), \\
\forall x \in \partial \Omega,\quad & u_{\epsilon}(x) - u(x) =
\displaystyle \sum_{n=1}^N \epsilon^n (I + \mathcal{E})^{-1} \,
[v^N_n](x)+ o(\epsilon^N).
\end{array}
\end{displaymath}
The remainder $o(\epsilon^N)$ depends only on $N$, $R$ and
$||h||_{\mathcal{C}^2}$.
\end{prop}

We can now compute the first order term in the expansion of
$\widetilde{c_{\mathrm{flr}}}$.

Recall that $\widetilde{c_{\mathrm{flr}}} = \delta
\,[u_{\epsilon}]\big|_{\partial C_{\epsilon}}$. The boundary
conditions (\ref{eq:u2}), that $u_{\epsilon}$ satisfies, give us
\begin{displaymath}
\widetilde{c_{\mathrm{flr}}} = \delta \beta \displaystyle
\frac{\partial u_{\epsilon}}{\partial \nu} = - \delta
\Psi_{\epsilon}.
\end{displaymath}

Let us find the first order approximation of $\Psi_{\epsilon}$. We
apply the previous procedure to obtain $\Psi^{(1)}_1$. Hence, one
introduces:
\begin{equation}\label{eq:psi0}
\begin{array}{l}
\vspace{0.3cm} \Psi^{(0)}_0 = -\beta (I + \beta \mathcal{L})^{-1}
\left[\displaystyle \frac{\partial H}{\partial \nu} \right], \\
\Psi^{(1)}_0 = - \beta (I + \beta \mathcal{L})^{-1} \left[ -
\displaystyle \frac{h'}{R^2}\frac{\partial H}{\partial \theta} + h
\frac{\partial^2 H}{\partial r^2} - h \displaystyle
\frac{\partial}{\partial
r}\mathcal{D}^{(0)}_{C}(\Psi^{(0)}_0)\right].
\end{array}
\end{equation}
Observe that $\Psi^{(0)}_0 = \Psi$. Thanks to (\ref{eq:v1}), one
can write $v_1^0$ for all $x \in
\partial \Omega$:
\begin{equation}\label{eq:v10}
 v_1^0(x) = - \displaystyle  \frac{1}{R}
  \int_0^{2\pi} \displaystyle \frac{\partial}{\partial
  {\theta^y}}\left(h(\theta^y) \displaystyle \frac{\partial}{\partial \theta^y}
  \Gamma^{(0)}(x,y)\right) \Psi(\theta^y) d\theta^y + \mathcal{D}^{(0)}_{C}[\Psi^{(1)}_0](x).
 \end{equation}
 Therefore, we get
 \begin{equation}\label{eq:psi1}
\begin{array}{l}
\vspace{0.3cm} \Psi^{(0)}_1 = \Psi^{(0)}_0 = \Psi\\
\Psi^{(1)}_1 = - \beta (I + \beta \mathcal{L})^{-1} \left( -
\displaystyle \frac{h'}{R^2}\frac{\partial H}{\partial \theta} + h
\frac{\partial^2 H}{\partial r^2} + \frac{\partial}{\partial r}
\mathcal{D}^{(0)}_{\Omega}\,(I + \mathcal{E})^{-1}[v^0_1] - h
\displaystyle \frac{\partial}{\partial r}
\mathcal{D}^{(0)}_{C}[\Psi]\right).
\end{array}
 \end{equation}

We first recall the mapping properties of the operators
$\mathcal{K}^{(0)}_D$ and $(\mathcal{K}^{(0)}_D)^*$. It is known
that if $D$ is a $\mathcal{C}^{2, \eta}$ domain, then
$\mathcal{K}^{(0)}_D$ and $(\mathcal{K}^{(0)}_D)^*$ map
continuously $\mathcal{C}^{1, \eta}(\partial D)$ into
$\mathcal{C}^{2, \eta}(\partial D)$ (see, for instance,
\cite{stein}). We also need the following result.
\begin{lem} \label{prop437}
Let $D$ be a $\mathcal{C}^{2, \eta}$ domain in $\mathbb{R}^2$, for
$0<\eta<1$. Let $\Psi \in \mathcal{C}^{1, \eta}(\partial D)$. We
have
\begin{displaymath}
\displaystyle \frac{\partial}{\partial
T}\mathcal{D}^{(0)}_D[\Psi]\bigg|_{\pm} = \mp\,\frac{1}{2}\,
\frac{\partial \Psi}{\partial T} + \displaystyle
\frac{\partial}{\partial T} \,\mathcal{K}^{(0)}_D[\Psi].
\end{displaymath}
\end{lem}
\begin{proof}
Let $\Psi \in \mathcal{C}^{1, \eta}(\partial D)$. Recall the jump
relation of the double layer potential across the boundary
$\partial D$:
\begin{displaymath}
\mathcal{D}^{(0)}_D [\Psi]|_{\pm} = \displaystyle \left( \mp
\frac{I}{2} + \mathcal{K}^{(0)}_D \right)[\Psi].
\end{displaymath}
The result of the proposition is simply obtained by taking the
tangential derivative of the previous formula and making use of
the mapping properties of $\mathcal{K}^{(0)}_D$.
\end{proof}

\begin{cor} \label{prop438}
Let $D$ be a $\mathcal{C}^{2, \eta}$ domain in $\mathbb{R}^2$, for
$0<\eta<1$. Let $h\in \mathcal{C}^2(\partial D)$ and let $\Psi \in
\mathcal{C}^{2, \eta}(\partial D)$. We have
\begin{equation}\label{eq:dDdT}
  - \displaystyle \frac{\partial}{\partial T}h \frac{\partial }{\partial T}
  \mathcal{D}^{(0)}_D[\Psi]
  \bigg|_{-} + \left (- \frac{I}{2} + (\mathcal{K}^{(0)}_D)^* \right )\!\!\left [  - \displaystyle
  \frac{\partial}{\partial T} h \frac{\partial \Psi}{\partial T} \right ] = \displaystyle \frac{\partial}{\partial T}
  \mathcal{K}^{(0)}_D\big[h\frac{\partial \Psi}{\partial T}\big] - \frac{\partial}{\partial T} h
  \frac{\partial}{\partial T}
  \mathcal{K}^{(0)}_D[\Psi].
 \end{equation}
 In the particular case of the disk $C$, we obtain that
 \begin{displaymath}
  - \displaystyle \frac{1}{R^2}\frac{\partial}{\partial \theta}\, h \frac{\partial }{\partial \theta}
  \,\mathcal{D}^{(0)}_C[\Psi] \bigg|_{-} + \left (- \frac{I}{2} + (\mathcal{K}^{(0)}_C)^*
   \right )\!\!\left [  - \displaystyle \frac{1}{R^2}\frac{\partial}{\partial \theta} h
   \frac{\partial \Psi}{\partial \theta} \right ] = 0.
 \end{displaymath}

\end{cor}
\begin{proof}
From Lemma \ref{prop437}, we know that
\begin{displaymath}
- \displaystyle \frac{\partial}{\partial T} \,h\, \frac{\partial
}{\partial T} \mathcal{D}^{(0)}_D[\Psi] \bigg|_{-} = - \frac{1}{2}
\displaystyle \frac{\partial}{\partial T} \,h \,\frac{\partial
\Psi}{\partial T} - \frac{\partial}{\partial T} \,h\,
\frac{\partial}{\partial T}\, \mathcal{K}^{(0)}_D[\Psi].
\end{displaymath}
Besides, the tangential derivative of the operator
$\mathcal{K}^{(0)}_D$ can be expressed as follows
\cite[p.144]{kress}
\begin{displaymath}
\displaystyle \frac{\partial}{\partial T}
\mathcal{K}^{(0)}_D[\Psi] = -
(\mathcal{K}^{(0)}_D)^*[\frac{\partial \Psi}{\partial T} ],
\end{displaymath}
for $\Psi \in \mathcal{C}^{2, \eta}(\partial D)$. We thus obtain
easily the result (\ref{eq:dDdT}).

Recall that, for a disk of radius $R$, the operator
$\mathcal{K}^{(0)}_C$ admits the explicit formula:
\begin{displaymath}
\mathcal{K}^{(0)}_C [\Psi] = \displaystyle \frac{1}{4 \pi}
\int_0^{2 \pi} \Psi(\phi) d\phi,
\end{displaymath}
which does not depend on $\theta$. Its tangential derivative is
therefore zero, and we have the formula for the disk. Finally, we
note that $(\mathcal{K}^{(0)}_C)^*=\mathcal{K}^{(0)}_C$ and hence,
$$
(\mathcal{K}^{(0)}_C)^*\!\!\left [\frac{\partial}{\partial \theta}
h
   \frac{\partial \Psi}{\partial \theta} \right ] = 0.
$$
\end{proof}

The next step is to find $w$ such that
 \begin{equation}\label{eq:defw}
 (I + \mathcal{E})[w] =  v_1^0.
 \end{equation}
 From Proposition \ref{prop434}, it follows that there exists a unique function $w$ solution to (\ref{eq:defw}).
 The following result holds.
 \begin{prop} \label{prop436}
 The solution to (\ref{eq:defw}) verifies the following equation and boundary conditions:
\begin{equation}\label{eq:w}
\left\{
\begin{array}{ll}
\vspace{0.25 cm}\Delta w = 0 & \textrm{\rm in}\, C \cup \Omega \setminus \overline{C},\\
\vspace{0.25 cm}\displaystyle\frac{\partial w}{\partial \nu}\bigg|_{+} - \frac{\partial w}{\partial \nu}
\bigg|_{-} = - \displaystyle \frac{1}{R^2}\frac{\partial}{\partial \theta} h \frac{\partial \Psi}{\partial
\theta} & \textrm{\rm on}\, \partial C,\\
\vspace{0.25 cm}w\mid_{+} - w\mid_{-} -\beta \displaystyle \frac{\partial w}{\partial \nu}
\bigg|_{-}= - \beta \, \bigg( \frac{h}{R}\frac{\partial u}{\partial r} +
\displaystyle \frac{1}{R^2}\frac{\partial}{\partial \theta} h \frac{\partial u}{\partial \theta}
\bigg|_{-}\bigg)& \textrm{\rm on} \,\partial C,\\
\displaystyle \frac{\partial w}{\partial \nu}\bigg|_{\partial
\Omega} = 0 & \textrm{\rm on}\, \partial \Omega .
\end{array}
\right.
\end{equation}
\end{prop}

\begin{proof}
The solution $w$ of the problem (\ref{eq:w}) satisfies the
representation formula:
\begin{equation} \label{eq:decw}
\forall x \in \Omega,  \quad w(x) =
\mathcal{D}^{(0)}_{\Omega}[w|_{\partial \Omega}](x) +
\mathcal{S}^{(0)}_C[ - \displaystyle
\frac{1}{R^2}\frac{\partial}{\partial \theta} h \frac{\partial
\Psi}{\partial \theta}](x) + \mathcal{D}^{(0)}_{C}[\Lambda](x),
\end{equation}
where the density $\Lambda$ on $\partial C$ is given by
\begin{equation}\label{eq:lambda}
\begin{array}{ll}
\vspace{0.3 cm} \Lambda = - \beta (I + \beta \mathcal{L})^{-1}
\left [- \displaystyle \frac{h}{R}\frac{\partial u}{\partial r}
\right. & - \displaystyle \frac{1}{R^2}\frac{\partial}{\partial
\theta} h \frac{\partial u}{\partial \theta} \bigg|_{-}
+ \frac{\partial}{\partial \nu} \mathcal{D}^{(0)}_{\Omega}[w|_{\partial \Omega}] \\
&\qquad \quad \left. + \left (- \displaystyle \frac{I}{2} +
(\mathcal{K}^* _C)^{(0)}\right )\left[  - \displaystyle
\frac{1}{R^2}\frac{\partial}{\partial \theta} h \frac{\partial
\Psi}{\partial \theta} \right]\right ].
\end{array}
\end{equation}
Thus, for $x \in \partial \Omega$,
\begin{equation}\label{eq:verif}
\begin{array}{l}
\vspace{0.25 cm} (I+\mathcal{E})[w](x) =  \mathcal{S}^{(0)}_C[ -
\displaystyle
\frac{1}{R^2}\frac{\partial}{\partial \theta} h \frac{\partial \Psi}{\partial \theta}](x) \\
\vspace{0.2cm}\qquad- \beta\, \mathcal{D}^{(0)}_C (I + \beta
\mathcal{L})^{-1}\left [-\displaystyle \frac{h}{R}\frac{\partial
u}{\partial r} - \displaystyle
\frac{1}{R^2}\frac{\partial}{\partial \theta} h \frac{\partial
u}{\partial \theta} \bigg|_{-} +\left (- \frac{I}{2} +
(\mathcal{K}^{(0)}_C)^* \right )\left[  - \displaystyle
\frac{1}{R^2}\frac{\partial}{\partial \theta} h \frac{\partial
\Psi}{\partial \theta} \right] \right ].
\end{array}
\end{equation}
By integrating by parts twice, the first term in our equation
becomes:
\begin{equation}\label{eq:first}
 \mathcal{S}^{(0)}_C[- \displaystyle \frac{1}{R^2}\frac{\partial}{\partial \theta} h \frac{\partial \Psi}{\partial \theta}]
 (x) = - \displaystyle  \frac{1}{R} \int_0^{2\pi} \displaystyle
 \frac{\partial}{\partial {\theta^y}}\left(h(\theta^y) \displaystyle
 \frac{\partial}{\partial \theta^y} \Gamma^{(0)}(x,y)\right) \Psi(\theta^y)
 d\theta^y.
\end{equation}
Hence, we obtain that
\begin{equation}\label{eq:first2}
 \mathcal{S}^{(0)}_C[- \displaystyle \frac{1}{R^2}\frac{\partial}{\partial \theta} h \frac{\partial \Psi}{\partial \theta}]
 (x) = v^1_1(x) - \mathcal{D}^{(0)}_C[\Psi^{(1)}_1](x).
\end{equation}
The representation formula of $u$ and the expression of the
Laplacian in terms of polar coordinates give us
\begin{equation}\label{eq:21}
\displaystyle \frac{1}{R^2}\frac{\partial}{\partial \theta} h
\frac{\partial u}{\partial \theta} \bigg|_{-} =  \displaystyle
\frac{h'}{R^2}\frac{\partial H}{\partial \theta}- h
\frac{\partial^2 H}{\partial r^2} - \displaystyle
\frac{h}{R}\frac{\partial H}{\partial r} + \displaystyle
\frac{1}{R^2}\frac{\partial}{\partial \theta} h \frac{\partial
}{\partial \theta} \mathcal{D}^{(0)}_C[\Psi] \bigg|_{-}.
\end{equation}
Observe that by definition of $\Psi$, we have on $\partial
\Omega$:
\begin{equation}\label{eq:22}
\displaystyle \frac{\partial u}{\partial r} = - \beta^{-1} \Psi.
\end{equation}
One can then derive the integral equation that  $\Psi$ verifies
and obtain that
\begin{equation}\label{eq:23}
-\displaystyle \frac{h}{R}\frac{\partial u}{\partial r} +
\displaystyle \frac{h}{R}\frac{\partial H}{\partial r} = -
\displaystyle \frac{h}{R}\frac{\partial}{\partial r}
\mathcal{D}^{(0)}_C[\Psi].
\end{equation}
The second term in our equation (\ref{eq:verif}) becomes
\begin{displaymath}
- \beta \mathcal{D}^{(0)}_C (I + \beta \mathcal{L})^{-1}\!\! \left
[ - \displaystyle \frac{h'}{R^2}\frac{\partial H}{\partial \theta}
+ h \frac{\partial^2 H}{\partial r^2} - \displaystyle
\frac{h}{R}\frac{\partial}{\partial r} \mathcal{D}^{(0)}_C(\Psi)
\right ].
\end{displaymath}
It follows from (\ref{eq:psi1}) and (\ref{eq:first2}) that
$$
\forall x \in \partial \Omega , \qquad (I + \mathcal{E})[w](x) =
v_1^0(x).
$$
\end{proof}
We have obtained an approximation at first order of
$\widetilde{c_{\mathrm{flr}}}$:
\begin{displaymath}
\widetilde{c_{\mathrm{flr}}} = c_{\mathrm{flr}} - \epsilon \delta
\Psi^{(1)}_1+ o(\epsilon),
\end{displaymath}
 where $\Psi^{(1)}_1$ is given by
  \begin{displaymath}
 \Psi^{(1)}_1 = - \beta (I + \beta \mathcal{L})^{-1} \left[ - \displaystyle \frac{h'}{R^2}\frac{\partial H}{\partial \theta}
  + h \frac{\partial^2 H}{\partial r^2} + \frac{\partial}{\partial r}
  \mathcal{D}^{(0)}_{\Omega} w - h \displaystyle \frac{\partial}{\partial r} \mathcal{D}^{(0)}_{C}(\Psi)\right],
 \end{displaymath}

 \noindent and $w$ is the solution of (\ref{eq:w}).

We can now derive the first order term in the asymptotic expansion
of (\ref{eq:int}) as $\epsilon\rightarrow 0$.

\begin{thm} \label{prop439} The integral (\ref{eq:int}) admits the following asymptotic expansion:
\begin{equation}\label{eq:int1}
\begin{array}{l}
\vspace{0.3 cm} \displaystyle \int_{\partial C_{\epsilon}}
\tilde{\gamma} \widetilde{c_{\mathrm{flr}}}(x)
\Phi_{\mathrm{exc}}^n(x)\Phi_{\mathrm{exc}}^m(x) ds(x) =
\displaystyle \int_{\partial C} \tilde{\gamma} c_{\mathrm{flr}}(x)
\Phi_{\mathrm{exc}}^n(x)\Phi_{\mathrm{exc}}^m(x) ds(x)\\
\qquad \qquad \qquad + \,\epsilon \displaystyle \int_{\partial
C}\tilde{\gamma} \left( A\,c_{\mathrm{flr}}(\theta)\,h(\theta) -
\delta \, B\, \Psi^{(1)}_1(\theta)\right)e^{-i(n + m)\theta}
\,d\theta + o(\epsilon),
\end{array}
\end{equation}
where the constants $A$ and $B$ are given by
\begin{equation}
\begin{array}{l}
\vspace{0.3 cm} A = ik J_n'(ikR) J_m(ikR)\,R + ik J_n(ikR) J_m'(ikR) \,R + J_n(ikR) J_m(ikR),\\
B =J_n(ikR) J_m(ikR) \,R.
\end{array}
\end{equation}
\end{thm}

\subsubsection{Fourier coefficients of $\Psi^{(1)}_1$}

Recall that $\Omega$ is the unit disk and $C$ is the disk with
radius $R<1$. In terms of polar coordinates, the fundamental
solution $\Gamma^{(0)}$ of $\Delta$ in $\mathbb{R}^2$, given by
(\ref{gamma0}), has the expression:
\begin{displaymath}
 \forall y\,(r,\theta) \in \overline{\Omega},\, \forall z\,(R,\phi) \in
 \overline{\Omega},
 \qquad \Gamma_z^0(y) = \displaystyle \frac{1}{4 \pi} \log(R^2 + r^2 - 2 r R \cos(\theta - \phi)).
\end{displaymath}
The decomposition of $\log$ into a power series gives us the
following formulas:
\begin{equation}\label{eq:gamma0}
\vspace{0.3 cm} \Gamma_z^0(y) = \left\{\begin{array}{ll}
\displaystyle \frac{1}{2 \pi}\log R - \displaystyle \frac{1}{4
\pi} \sum_{n \in \mathbb{Z}^*} \frac{1}{|n|}
(\frac{r}{R})^{|n|} \,e^{in(\theta - \phi)} \quad & \textrm{if}\, r < R, \\
 \nm  \displaystyle \frac{1}{2 \pi}\log r - \displaystyle \frac{1}{4 \pi}
  \sum_{n \in \mathbb{Z}^*} \frac{1}{|n|} (\frac{R}{r})^{|n|}
  \,e^{in(\theta - \phi)} \quad & \textrm{if}\, R < r .
\end{array}
\right.
\end{equation}
Let $f \in L^2(]0, 2\pi[)$. By reinjecting (\ref{eq:gamma0}) into
the definition of the following operators, we obtain for $y(R,
\theta) \in \partial C$ that
\begin{displaymath}
\begin{array}{rl}
\vspace{0.3 cm} \mathcal{S}^{(0)}_{\Omega}[f](y) & = -
\displaystyle \frac{1}{2} \sum_{n \in \mathbb{Z}^*}
 \frac{1}{|n|} R^{|n|}\,\hat{f}(n) \,e^{in\theta},\\
\vspace{0.3 cm} \mathcal{D}^{(0)}_{\Omega}[f](y) & = \hat{f}(0) +
\displaystyle \frac{1}{2} \displaystyle  \sum_{n \in \mathbb{Z}^*} R^{|n|}\,\hat{f}(n) \,e^{in\theta},\\
\vspace{0.3 cm} \displaystyle \frac{\partial
\mathcal{D}^{(0)}_{\Omega}}{\partial r}[f](y) &
= \displaystyle \frac{1}{2}\displaystyle  \sum_{n \in \mathbb{Z}^*} |n|R^{|n| - 1}\,\hat{f}(n) \,e^{in\theta},\\
\displaystyle \frac{\partial \mathcal{D}^{(0)}_{C}}{\partial
r}[f](y) & = \displaystyle \frac{1}{2}\displaystyle
 \sum_{n \in \mathbb{Z}^*} |n|\,\frac{1}{R} \,\hat{f}(n) \,e^{in\theta} .\\
\end{array}
\end{displaymath}
Recall that $H$ satisfies the following representation formula on
$\partial C$:
\begin{displaymath}
H = - \mathcal{S}^{(0)}_{\Omega}[g_{\mathrm{ele}}] +
\mathcal{D}^{(0)}_{\Omega}[f_0],
\end{displaymath}
where $g_{\mathrm{ele}} = \displaystyle \frac{\partial u}{\partial
\nu} \bigg|_{\partial \Omega}$ and $f_0 = u|_{\partial \Omega}$.
We therefore get
\begin{displaymath}
\begin{array}{cl}
\vspace{0.3 cm} H(\theta) &=  \hat{f_0}(0) + \displaystyle
\frac{1}{2} \displaystyle\sum_{n \in \mathbb{Z}^*}
\left(\frac{1}{|n|}\hat{g}_{\mathrm{ele}}(n) +
\hat{f_0}(n) \right ) R^{|n|}\, e^{in\theta},\\
\vspace{0.3 cm} \displaystyle \frac{\partial H}{\partial \theta}(\theta)
&= \displaystyle \sum_{n \in \mathbb{Z}^*} in \,\widehat{H}(n)\, e^{in\theta},\\
 \displaystyle \frac{\partial^2 H}{\partial r^2}(\theta) &= \displaystyle  \frac{1}{R^2}
 \sum_{n \in \mathbb{Z}^*} |n|(|n| -1) \widehat{H}(n) \,
 e^{in\theta}.
\end{array}
\end{displaymath}
Besides, for $f \in \mathcal{C}^{2, \eta}(\partial C)$,  we have
\begin{displaymath}
(I + \beta \mathcal{L})^{ -1}[f](\theta) =  \displaystyle  \sum_{n
\in \mathbb{Z}^*}\left( 1 + \beta\, \frac{|n|}{2R} \right ) ^{-1}
\hat{f}(n)\, e^{in\theta}.
\end{displaymath}
Note that $\widehat{\Psi}(n) = - \beta \displaystyle \left( 1 +
\beta \,\frac{|n|}{2R} \right ) ^{-1} \frac{|n|}{R}\,
\widehat{H}(n)$.

We can now write the Fourier coefficients of $\Psi^{(1)}_1\,$, for
$n \in \mathbb{Z}^{*} :=\{ m \in \mathbb{Z}, m \neq 0\}$,
\begin{equation}\label{eq:fourierpsi11}
\begin{array}{l}
\vspace{0.3 cm}\widehat{\Psi^{(1)}_1}(n) =- \beta\,\displaystyle \frac{1}{2} \displaystyle \frac{|n| \, R^{|n| -1}}{1 + \beta \, \displaystyle \frac{|n|}{2R}}\, \hat{w}(n) - \beta \, \displaystyle \sum_{p = - \infty}^{\infty} \hat{h}(p) \widehat{H}(n-p) \\
 \times\displaystyle \left((n - p)p + |n-p|(|n-p| - 1)+ \frac{\beta}{R} \frac{|n-p|^2}{2 R
 + \beta \, |n-p|}\right) \left(1 + \beta \, \frac{|n|}{2
 R}\right)^{-1}.
\end{array}
\end{equation}

Integral (\ref{eq:int}) becomes at first order:
\begin{displaymath}
\begin{array}{rl}
\vspace{0.3 cm} \mathcal{I} _{\epsilon}^{m,n} = \mathcal{I}_0^{m,n} &+\epsilon\, \,2 \pi \,A \,\delta \,\beta \tilde{\gamma} \displaystyle \sum_{p = - \infty}^{\infty}\hat{h}(p)\widehat{H}(m+n-p) \left(1 + \beta \, \displaystyle\frac{|m+n-p|}{2 R}\right)^{-1}\\
 &- \epsilon \,\,2 \pi\, B \,\delta\,\tilde{\gamma}\, \widehat{\Psi^{(1)}_1}(m+n),
\end{array}
\end{displaymath}
where $\mathcal{I}_{\epsilon}^{m,n} = \displaystyle \int_{\partial
C_{\epsilon}} \tilde{\gamma} \widetilde{c_{\mathrm{flr}}}(x)
\Phi_{\mathrm{exc}}^n(x)\Phi_{\mathrm{exc}}^m(x) \, ds(x)$ and
$\mathcal{I}_0^{m,n} = \displaystyle \int_{\partial
C}\tilde{\gamma} c_{\mathrm{flr}}(x)
\Phi_{\mathrm{exc}}^n(x)\Phi_{\mathrm{exc}}^m(x) ds(x)$.

\subsubsection{Reconstruction of $h$}

We introduce the linear operator $\mathcal{Q}$ defined on
$\mathcal{C}^2(\partial C)$ by
\begin{displaymath}
(\mathcal{Q}[\hat{h}])_{m,n} = \epsilon \displaystyle \sum_{p = -
\infty}^{\infty}F_{m,n}(p) \,\hat{h}(p),
\end{displaymath}
where
\begin{displaymath}
\begin{array}{r}
 \vspace{0.2 cm} F_{m,n}(p)  = 2 \pi\delta\beta \tilde{\gamma} \!\left[\displaystyle \frac{A}{1 + \beta \, \displaystyle\frac{|m+n-p|}{2 R}} + \displaystyle\frac{B}{1 + \beta \, \displaystyle\frac{|m+n|}{2 R}}  \bigg((m+n - p)p  \right. \qquad \qquad\qquad\quad\\
 \left. + |m+n-p|(|m+n-p| - 1)+\displaystyle \frac{\beta}{R} \frac{|m+n-p|^2}{2 R + \beta \,
 |m+n-p|}\right)\Bigg] \widehat{H}(m + n - p).
 \end{array}
 \end{displaymath}
 Recall that $\mathcal{I}_{\epsilon}^{m,n}$ and $\mathcal{I}_{0}^{m,n}$ can be computed from the
 knowledge of the outgoing light intensities
  $I^n_{\mathrm{emt}, \epsilon}$ and $I^n_{\mathrm{emt}}$ measured at the
  boundary of our domain (\ref{eq:iepsilonm}),
 (\ref{eq:im}):
 \begin{displaymath}
 \mathcal{I}_{\epsilon}^{m, n} = {2 \pi} \,E_m \, \widehat{I_{\mathrm{emt}, \epsilon}^n}(m),
 \qquad  \mathcal{I}_0^{m, n} = {2 \pi} \,E_m\,
 \widehat{I_{\mathrm{emt}}^n}(m).
 \end{displaymath}
 We denote $\hat{a}$ the data of our problem:
\begin{displaymath}
\forall m, n \in \mathbb{Z}, \qquad  \hat{a}_{m, n} : = {2 \pi}
E_m \!\! \left(\! \widehat{I_{\mathrm{emt}, \epsilon}^n}(m) -
\widehat{I_{\mathrm{emt}}^n}(m)\!\right) \!- \epsilon\,
\tilde{\gamma} B \, \beta\, \pi \, \delta\, \displaystyle \frac{|m
+ n| \, R^{|m + n| -1}}{1 + \beta\displaystyle \frac{ |m +
n|}{2R}}\, \hat{w}(m + n),
\end{displaymath}
 where $\epsilon w$ is the measured difference of the voltage potential on
 $\partial \Omega$, when the cell occupies $C_{\epsilon}$ and when it is the circle $C$.

The operator $\mathcal{Q}$ links the perturbation $h$ of the
membrane cell to the data of our problem:
\begin{displaymath}
\hat{a}_{m, n} = (\mathcal{Q}[\hat{h}])_{m,n} + \epsilon^2
\hat{V}_{m, n},
\end{displaymath}
 with the term $\vspace{0.2cm} \epsilon^2 \hat{V}_{m, n}$ modeling the linearization error.

 We choose to apply at the boundary of our domain $\Omega$ an electric field $g_{\mathrm{ele}}:\theta \to
 e^{iz\theta}$ with $z\in \mathbb{Z}$. Let us compute the resulting voltage potential at
 the boundary of $\Omega$, $f_0$ and more specifically its Fourier coefficients. From the representation
 formula (\ref{eq:decu}) of $u$  and the jumps relation of the single and double layer potentials, we obtain the following equation
 at the boundary of our domain:
 \begin{displaymath}
f_0= - \mathcal{S}^{(0)}_{\Omega}[g_{\mathrm{ele}}] +
\displaystyle \frac{1}{2} f_0 + \mathcal{K}^{(0)}_{\Omega}[f_0] +
\mathcal{D}^{(0)}_C[\Psi].
 \end{displaymath}

Since $\displaystyle \int_{\partial \Omega} f_0 = 0$ from
(\ref{eq:u2}), we immediately get $\hat{f_0}(0) = 0$ and
$\mathcal{K}^{(0)}_{\Omega}[f_0]=0$. We write, like in the
previous section, the Fourier coefficients of the various layer
potentials and of $\Psi$ and get for $n \in \mathbb{Z} \setminus
\{0\}$:
\begin{displaymath}
\hat{f_0}(n) = \displaystyle \frac{2(1 + \beta\,
\displaystyle\frac{|n|}{2R}) + \beta |n| R^{2|n|-2}}{2(1 + \beta
\,\displaystyle\frac{|n|}{2R}) - \beta |n|
R^{2|n|-2}}\,\frac{1}{|n|}\,\hat{g}_{\mathrm{ele}}(n).
\end{displaymath}
Note that $\hat{g}_{\mathrm{ele}}(n)=\delta_z(n)$. We can now
write the Fourier coefficients of $H|_{\partial C}$ in our case:
\begin{displaymath}
\widehat{H}(0) = 0,\qquad \textrm{and} \qquad \forall n \in
\mathbb{Z}\setminus \{0\}, \, \, \widehat{H}(n) = \displaystyle
\frac{2(1 + \beta\, \displaystyle\frac{|n|}{2R})}{2(1 + \beta
\,\displaystyle\frac{|n|}{2R}) - \beta |n|
R^{2|n|-2}}\,\frac{1}{|n|}\,\delta_z(n) R^{|z|}.
\end{displaymath}

The operator $\mathcal{Q}$ has therefore the following simplified
expression:
\begin{displaymath}
(\mathcal{Q}[\hat{h}])_{m,n} = \epsilon \, F_{m,n}(z)
\,\hat{h}(m+n-z),
\end{displaymath}
where
\begin{displaymath}
\begin{array}{r}
 \vspace{0.2 cm} F_{m,n}(z) =\left[\displaystyle \frac{A}{1 + \beta \, \displaystyle\frac{|z|}{2 R}} +
 \displaystyle\frac{B}{1 + \beta \, \displaystyle\frac{|m+n|}{2 R}}  \bigg((m+n-z)z  \displaystyle+
 |z|(|z| - 1)+\displaystyle \frac{\beta}{R} \frac{|z|^2}{2 R + \beta \, |z|}\bigg)\right] \\
  \times \displaystyle \frac{ 2 \pi\delta\beta \tilde{\gamma}}{|z|} \,\displaystyle\frac{2(1 + \beta\,
  \displaystyle\frac{|z|}{2R})}{2(1 + \beta \,\displaystyle\frac{|z|}{2R}) -
  \beta |z| R^{2|z|-2}}\,R^{|z|}.\qquad \qquad
 \end{array}
 \end{displaymath}
 Recall that the constants $A$ and $B$ depend on $R$ and $k$.

The adjoint of the operator $\mathcal{Q}$ is given by
\begin{displaymath}
(\mathcal{Q}^{\star}[\hat{a}])_p = \epsilon \displaystyle \sum_{j
= - \infty}^{\infty} \overline{F}_{j,p+z-j}(z)\,
\hat{a}_{j,p+z-j}.
\end{displaymath}
Then we obtain that
\begin{displaymath}
(\mathcal{Q}^{\star} \mathcal{Q}[\hat{h}])_p = \epsilon^2
\displaystyle \sum_{j = - \infty}^{\infty} |F_{j, p+z-j}(z)|^2 \,
\hat{h}(p).
\end{displaymath}

We now consider the presence of measurement or instrument noise in
our measured data. We thus introduce:
\begin{displaymath}
\hat{a}_{m, n}^{meas} = (\mathcal{Q}[\hat{h}])_{m,n} + \epsilon^2
\hat{V}_{m, n} + \sigma \hat{W}_{m, n},
\end{displaymath}
with the noise term $\hat{W}_{m, n}$ modeled as independent
standard complex circularly symmetric Gaussian random variables
(such that $\mathbb{E}[|\hat{W}_{m, n}|^2] = 1$; $\mathbb{E}$
being the expectation). Here, $\sigma$ corresponds to the noise
magnitude. We consider that $\sigma$ verifies $\epsilon^2 \ll
\sigma$, so that the linearization error is negligible over the
measurement error and we can write:
\begin{displaymath}
\hat{a}_{m, n}^{meas} = (\mathcal{Q}[\hat{h}])_{m,n} + \sigma
\hat{W}_{m, n}.
\end{displaymath}
Following the methodology of \cite{19,20}, we want to asses the
resolving power of the measured data in the presence of this
noise.

Since $h$ is $\mathcal{C}^2$, $|\hat{h}(p)| \leq C/p^2$ for some
constant $C$, for all $p \in \mathbb{Z} \setminus \{0\}$. Besides,
one can see that for all $m, n \in \mathbb{Z}$, $F_{m,n}$ is
bounded, for given $R$ and $k$. Let $M$ be a positive real such
that $M\ll1/\epsilon^2$. We can reconstruct the Fourier
coefficients of the shape deformation $h$ only for $p$ such that
$|p| \leq M$, otherwise the linearization error $\epsilon^2
\hat{V}_{m,n}$ is too large. We suppose that $\hat{h}_p = 0$ for
all $|p| \geq M$.

To reconstruct $h$, one can minimize the following quadratic
functional over $\varphi$:
\begin{displaymath}
\bigg|\bigg|\mathcal{Q} [\hat{\varphi}] - \hat{a}^{meas}
\bigg|\bigg|_F^2,
\end{displaymath}
where $\hat{a}^{meas} = (\hat{a}^{meas}_{m,n})_{m,n}$,
$\hat{\varphi}=(\hat{\varphi}(p))_p$, and $|| \; ||_F$ is the
Frobenius norm. The obtained least squares estimate is given by
\begin{equation} \label{lsq}
\forall p \in \ [-M,M],  \quad \hat{h}_{est}(p) =
(\mathcal{Q}^{\star} \mathcal{Q})^{-1}
\mathcal{Q}^{\star}[\hat{a}^{meas}](p) = \hat{h}(p) + \sigma
\left((\mathcal{Q}^{\star} \mathcal{Q})^{-1} \mathcal{Q}^{\star}
[\hat{W}]\right)_p.
\end{equation}

One can prove with the explicit formulas of the operators
$\mathcal{Q}$ and $\mathcal{Q}^{\star}$ that the following result
holds.
\begin{prop} \label{proplsa}
Estimation (\ref{lsq}) is unbiased and has the following variance:
\begin{equation}\label{eq:var}
\mathbb{E}\left(|\hat{h}_{est}(p) -
\hat{h}(p)|^2\right)=\displaystyle \frac{\sigma^2}{\epsilon^2}
\left(  \sum_{j = - \infty}^{\infty} |F_{j, p+z-j}|^2\right)^{-1}.
\end{equation}
\end{prop}

Besides Proposition \ref{proplsa}, Parseval's identity and Graf's
addition formula yield
\begin{displaymath}
\sum_{j = - \infty}^{\infty} |F_{j, p+z-j}|^2 = \displaystyle
\frac{2} {\pi} \int_{0}^{\pi/2} |f_p(\theta)|^2 d\theta,
\end{displaymath}
where the function $f_p$ is defined by
$$f_p(\theta)=a \,2 i k R \sin(\theta) \,J'_{p+z}(2ikR\sin(\theta)) + (a+Rb) \,J_{p+z}(2ikR\sin(\theta)),$$
\begin{displaymath}
\begin{array}{ll}
\vspace{0.3 cm} \textrm{with}& a(R,z)= \displaystyle \frac{2R^{|z|}}{2(1 + \beta
 \,\displaystyle\frac{|z|}{2R}) - \beta |z| R^{2|z|-2}} \displaystyle \frac{ 2 \pi\delta\beta \tilde{\gamma}}{|z|}, \\
&b(R,p,z)= a(R,z)\, \displaystyle\frac{2R + \beta \, |p+z|}{2R +
\beta\, |z|}  \bigg(pz  \displaystyle+ |z|(|z| - 1)+\displaystyle
\frac{\beta}{R} \frac{|z|^2}{2 R + \beta \, |z|}\bigg).
\end{array}
\end{displaymath}

We introduce the signal to noise ratio $\textrm{SNR}$:
\begin{equation} \label{defsnr}
\textrm{SNR} = \displaystyle (\frac{\epsilon}{\sigma})^2.
\end{equation}
The following result holds thanks to (\ref{eq:var}).
\begin{thm} \label{thmresolving}
Suppose that the $p$th mode of $h$, $\hat{h}(p)$, is of order $1$,
we can resolve it if the following condition is satisfied:
\begin{displaymath}
\textrm{SNR}^{-1}<\displaystyle \frac{2} {\pi} \int_{0}^{\pi/2}
|f_p(\theta)|^2 d\theta.
\end{displaymath}
\end{thm}

Let us simplify this stability condition under the respective
asymptotic assumptions $|k|R\gg1$ and $|k|R \ll 1$.

Since $J_{-n}=(-1)^n J_n$ (\cite[Formula 9.1.5]{12}), we can
consider without any restriction that $p+z\geq 0$.

\paragraph{Assumption 1: $\mathbf{|k|R \gg1}$}
We assume in this paragraph that $|k|R \gg 1$. We use the
asymptotic expansions of the Bessel functions of the first kind
and their derivative (\cite[Formulas 9.2.5 and 9.2.11]{12}) to
find that, in this case, when $p+z<2 |k| R$, we have
\begin{displaymath}
\displaystyle \frac{2} {\pi} \int_{0}^{\pi/2} |f_p(\theta)|^2
d\theta \sim \displaystyle \frac{4a^2}{\pi^2} |k| R
\sum_{n=0}^{\infty} \frac{(4\mathit{Im}(ik)R)^{2n}}{(2n)!}
\frac{2^{2n}(n!)^2}{(2n+1)!}.
\end{displaymath}
Then the resolving condition becomes
\begin{displaymath}
\textrm{SNR}^{-1}< C(R,z) |k| \qquad \textrm{with}\quad C(R,z) =
\displaystyle \frac{4\,a(R,z)^2}{\pi^2} R  \sum_{n=0}^{\infty}
\frac{(4\mathit{Im}(ik)R)^{2n}}{(2n)!}
\frac{2^{2n}(n!)^2}{(2n+1)!}.
\end{displaymath}

With large $|k|R$, we can estimate the coefficients $\hat{h}(p)$
for all $\textrm{SNR}$ of order $1/|k|$, as long as $p+z<2 |k|R$.

When $p+z>2 |k| R$,  from \cite[Formulas (9.3.35) and
(9.3.43)]{12} it follows that the following asymptotic behavior of
our integrand holds:
\begin{displaymath}
|f_p(\theta)|^2 \sim \frac{\sqrt{\left |1-x\right|}}{2(p+z)\pi}\,
\left|1+\sqrt{1- x}\right |^{-(p+z)}\,
\textrm{e}^{2(p+z)\mathit{Re}(\sqrt{1- x})}\, x^{2(p+z-1)},
\end{displaymath}
where $\vspace{0.2cm} x=\displaystyle
\left(\frac{2ikR\sin(\theta)}{p+z}\right)^2$.

Since $|x|<1$, the last term in the preceding expression is the
dominant one, and makes the integral exponentially small. To
resolve the $p$th mode of $h$ in this context, we therefore need a
SNR exponentially large, which is impossible in practice.

We choose for each $p<M$ an electric model with $z<2 |k| R -p$.
The condition $p+z<2 |k| R$ is in this way always satisfied, and
the $p$th mode can be resolved as long as $\textrm{SNR}^{-1}<
C(R,z) |k|$.

For a fixed $z$, $k$ and $\textrm{SNR}$, this inequality gives us
a constraint on the cell radius. In order to be able to image the
cell with a given $\textrm{SNR}$, its radius has to be larger than
a minimal value, $R^{\star}$ given by
\begin{displaymath}
R^{\star}(\textrm{SNR}) = \mathcal{F}^{-1}(\textrm{SNR}^{-1}),
\end{displaymath}
with
\begin{displaymath}
\mathcal{F}(t)=\displaystyle \frac{4\,a(t,z)^2}{\pi^2} t |k|
\sum_{n=0}^{\infty} \frac{(4\,\mathit{Re}(k)t)^{2n}}{(2n)!}
\frac{2^{2n}(n!)^2}{(2n+1)!}.
\end{displaymath}

The typical size of eukaryotes cell is $10/100 \,{\mu}\textrm{m}$.
We use for our different parameters the following realistic values
reported in \cite{11}, \cite{3}, \cite{21}, \cite{9}:
\begin{itemize}
\item the absorption coefficient $\mu = 0.03$, \item the reduced
scattering coefficient $\mu'_s=0.275$, \item the fluorophore
quantum efficiency $\eta=0.016$, \item the fluorophore
fluorescence lifetime $\tau=0.56\, \textrm{s}^{-1}$, \item the
fluorophore extinction coefficient $\varepsilon_{\mathrm{exc}}
=5*10^4 \,\textrm{mm}^{-1}\textrm{mol}^{-1}$, \item The constant
$\delta$ defined in (\ref{eq:cf}) is given by $\delta =
0.91*10^{-6}\, \textrm{mol}\,\textrm{V}^{-1}$.
\end{itemize}

It is worth mentioning that the absorption coefficient $\mu$ is
low compared to the reduced scattering coefficient $\mu'_s$.
Recall that $k= \left(\displaystyle \frac{\mu + i
\omega/c}{D}\right)^{1/2}$. Then, for given absorption and reduced
scattering coefficients, Assumption $1$ corresponds to frequencies
$\omega$ such that $\vspace{0.2cm} \omega \gg 10^{16}$ and
therefore, are nonphysical. The minimal radius $R^{\star}$
increases with $z$, we thus choose $z$ such as $|z|=1$. Since $M
\sim 10$ with these values of the parameters, this choice does not
impose any restriction, because we have always $M-1<2|k|R$.

\paragraph{Assumption 2: $\mathbf{|k|R \ll1}$}
Note  that the larger the reduced scattering coefficient is, the
smaller is $|k|$. The asymptotic expansions of the Bessel
functions of the first kind and their derivative when the argument
tends to zero
 (\cite[Formula 9.1.7]{12}), give us the asymptotic behavior of our integral in the case
 of a small $|k|R$:
 \begin{displaymath}
 \displaystyle \frac{2} {\pi} \int_{0}^{\pi/2} |f_p(\theta)|^2 d\theta \sim \displaystyle
 \left(\frac{|k|R}{2}\right)^{2(p+z)} \frac{(2(p+z))!}{(p+z)!^4} \left(a(p+z+1)+Rb\right)^2.
 \end{displaymath}

For fixed $z$, $k$ and $R$, the $p$th mode of $h$ can be resolved
under  Assumption 2 as long as the $\textrm{SNR}$ verifies:
 \begin{displaymath}
 \textrm{SNR}^{-1}< \displaystyle \left(\frac{|k|R}{2}\right)^{2(p+z)} \frac{(2(p+z))!}{(p+z)!^4}
 \left(a(R,z)(p+z+1)+Rb(R,p,z)\right)^2.
 \end{displaymath}

 If we consider now that the $\textrm{SNR}$, $k$ and $z$ are given, we can define,
 for each mode $p$, the minimal resolving
 radius $R^{\star}$, {\it i.e.}, the smallest radius that the cell can
 have if we want to resolve the $p$th mode of its membrane deformation.
\begin{thm} \label{proplsa2}
 The minimal resolving radius $R^{\star}$ has
  the following expression:
 \begin{displaymath}
R^{\star}(\textrm{SNR},p) = \mathcal{F}_p^{-1}(\textrm{SNR}^{-1}),
\end{displaymath}
 where the function $\mathcal{F}_p$  in this regime is given by
 \begin{displaymath}
\mathcal{F}_p(t)=\displaystyle
\left(\frac{|k|t}{2}\right)^{2(p+z)} \frac{(2(p+z))!}{(p+z)!^4}
\left(a(t,z)\,(p+z+1)+t\,b(t,p,z)\right)^2.
\end{displaymath}
\end{thm}

Note that the higher the reduced scattering coefficient is, the
better is the resolving power of the imaging method. In fact, in
order to resolve the mode $p$, the higher the reduced scattering
coefficient is, the smaller is the required $\textrm{SNR}$.

We plot in Figure \ref{fig2} this minimal resolving radius as a
function of the $\textrm{SNR}$ for $p=0$, $1$, $2$ and $3$. We
centered the $y$-axis on the typical radii of eukaryotes cells,
like in the preceding paragraph. Assumption 2 corresponds to
frequencies $\omega$ such that $\omega \ll 10^{13}$. We choose
$\omega = 10^9$, which is a typical frequency used in cellular
tomography. For each $p$, we took $z=\delta_0(p) -p$, because
$R^{\star}$ decreases with $p+z$. Since we can not take $z=0$, the
mode $0$ is not the easiest to resolve. For the other parameters,
we kept the values of the previous paragraph.

\begin{figure}
\centering
\includegraphics[width=0.7\textwidth]{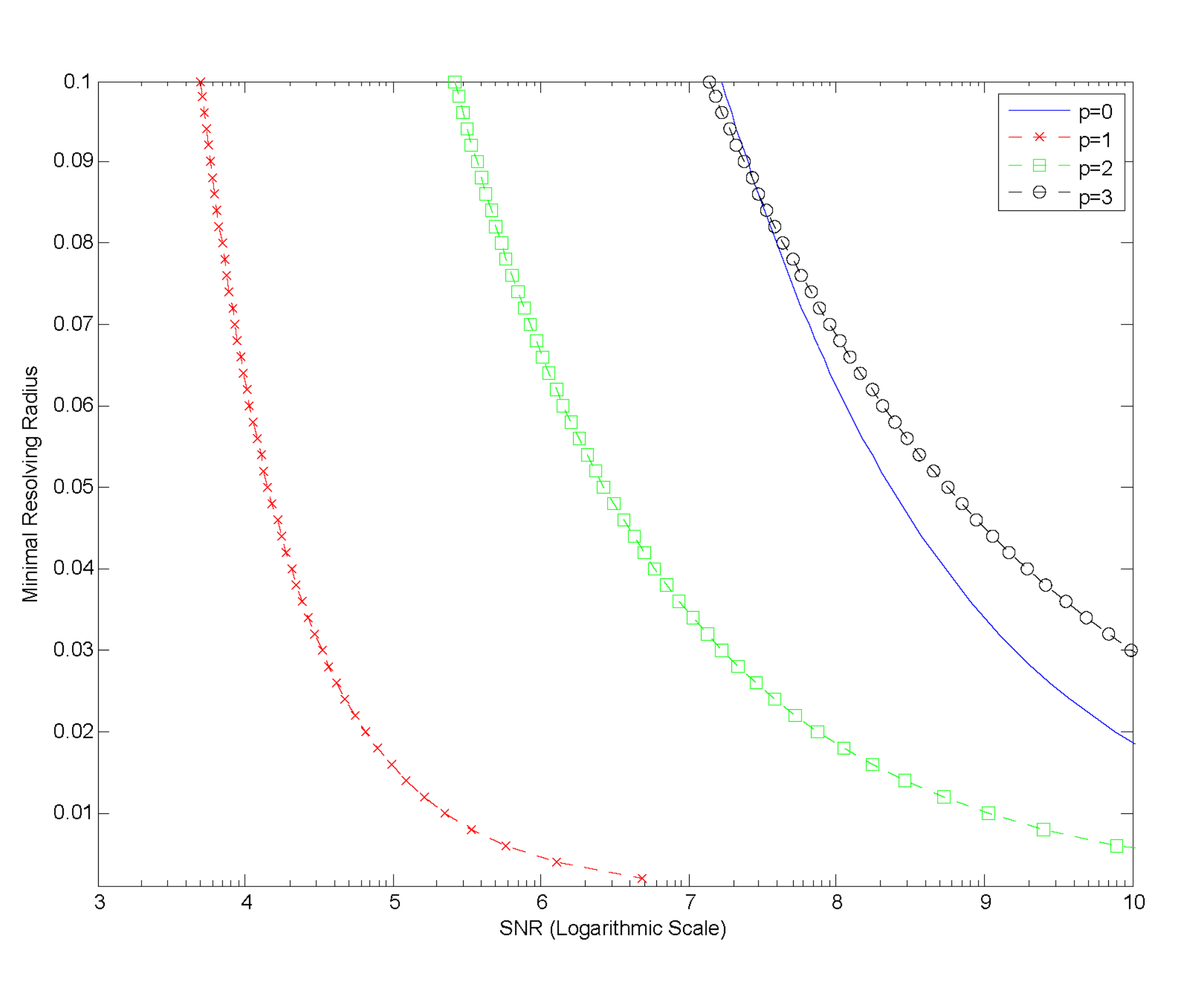}
 \caption{\it{Minimal resolving radius as function of the $\textrm{SNR}$ when $|k|R\ll1$.}\label{fig2}}
\end{figure}

Under Assumption 1, for given $z$, $R$ and $\textrm{SNR}$, if the
resolving condition was verified, we could resolve all modes of
$h$ up to $M$. Because the constraint depends this time on $p$, a
new question arises: "how many modes can we resolve for fixed $R$
and $\textrm{SNR}$?". We introduce the maximal mode number
$p(R,\textrm{SNR})$ defined by
\begin{displaymath}
p(R,\textrm{SNR})=\sup\left\{p'\in\mathbb{N}
\setminus\{0\}|\displaystyle \inf_{1\geq p' \geq
p}\mathcal{F}_{p'}(R)
> \textrm{SNR}^{-1}\right \} +
\mathbb{1}_{\mathcal{F}_0(R)>\textrm{SNR}^{-1}},
\end{displaymath}
which answers this question.

We plot in Figure \ref{fig3} the maximal mode number as a function
of the cell radius for different values of the $\textrm{SNR}$. We
took the same values of our parameters as in Figure \ref{fig2}.

\begin{figure}
\centering
\includegraphics[width=0.7\textwidth]{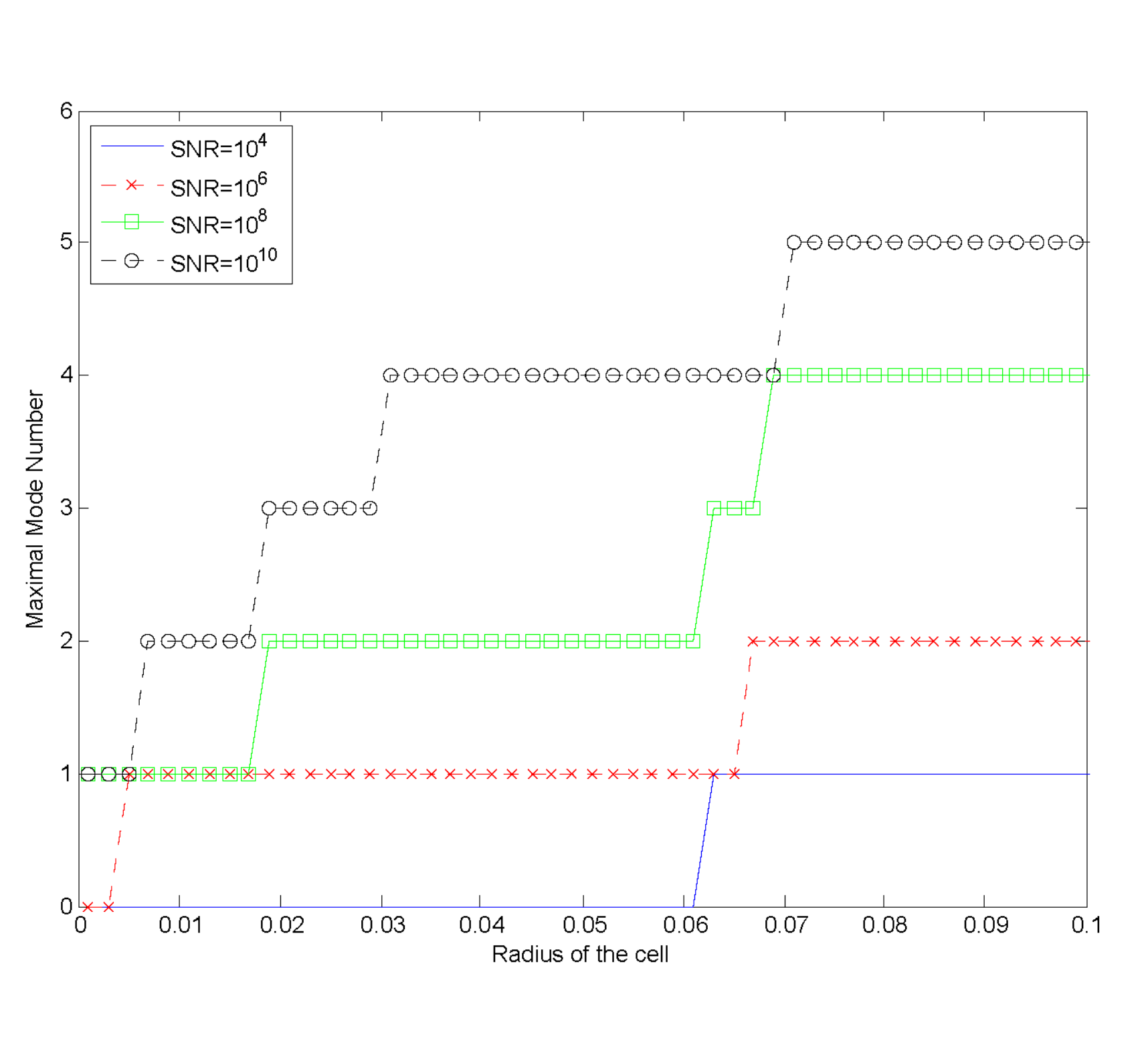}
 \caption{\it{Maximal Mode Number as function of the cell radius when $|k|R\ll1$.}\label{fig3}}
\end{figure}

\subsection{Reconstruction of the cell membrane in the general two-dimensional case}

We leave the specific case of a circular domain to go back to the
general case in dimension two. Let $a,b \in \mathbb{R}$ with
$a<b$. Let $x : [a,b] \to \mathbb{R}^2$ be a parametrization of
$\partial C$ such that $x \in \mathcal{C}^{2, \eta}(\mathbb{R})$
for an $\eta
>0$ and $|x'| = 1$. The outward unit normal to $\partial C$ at
$x(t)$, $\nu(x)$ and the tangential vector, $T(x)$, are given by
\begin{displaymath}
\nu(x) = R_{- \frac{\pi}{2}} x'(t), \qquad T(x) = x'(t),
\end{displaymath}
where $R_{-\frac{\pi}{2}}$ is rotation by $-\frac{\pi}{2}$.

We introduce the curvature $\tau$ defined for all  $x \in \partial
C$ by
\begin{displaymath}
x''(t) = \tau(x) \nu(x).
\end{displaymath}

Let $C_{\epsilon}$ be an $\epsilon$-perturbation of $C$, {\it
i.e.}, there is $h \in \mathcal{C}^2([a, b])$, such that $\partial
C_{\epsilon}$ is given by
\begin{displaymath}
\partial C_{\epsilon} = \left\{\displaystyle \tilde{x}; \tilde{x}(t) =
x(t) + \epsilon \,h(t)\nu(x(t)), t \in [a, b]\right\}.
\end{displaymath}

Like in the previous section, our goal is to reconstruct the shape
deformation $h$ of our cell. Let $I^{\,g}_{\mathrm{emt},\epsilon}$
(resp. $I_{\mathrm{emt}}^{\,g}$) be the outgoing light intensities
measured at the boundary of our domain when the cell occupies
$C_{\epsilon}$ (resp. $C$) and the optical source $g$ is applied
at $\partial \Omega$. It follows from Proposition \ref{prop411}
that
\begin{equation}\label{eq:im2}
\begin{array}{ll}
\vspace{0.3 cm}&\displaystyle \int_{\partial C_{\epsilon}}
\tilde{\gamma} \widetilde{c_{\mathrm{flr}}}(x)
\Phi_{\mathrm{exc}}^{\,f}(x)\Phi_{\mathrm{exc}}^{\,g}(x) ds(x) =
\int_{\partial \Omega}
f \,I^{\,g}_{\mathrm{emt},\epsilon}\, ds(x),\\
\textrm{resp.} \quad &\displaystyle \int_{\partial C}
\tilde{\gamma} c_{\mathrm{flr}}(x)
\Phi_{\mathrm{exc}}^{\,f}(x)\Phi_{\mathrm{exc}}^{\,g}(x) \, ds(x)
= \int_{\partial \Omega} f \,I^{\,g}_{\mathrm{emt}}\, ds(x),
\end{array}
\end{equation}
where $f,g \in L^2(\partial \Omega)$ and
$\widetilde{c_{\mathrm{flr}}}$ (resp. $c_{\mathrm{flr}}$) is the
concentration of fluorophores on the boundary of the cell
$\partial C_{\epsilon}$ (resp. $\partial C$).

We introduce the voltage potential $u$ such that $c_{\mathrm{flr}}
= \delta [u]|_{\partial C}$. We know, from Proposition
\ref{prop421}, that $u$ admits the following representation
formula:
\begin{displaymath}
\begin{array}{rccccc}
\forall x \in \Omega, & \quad u(x) &=& H(x) &+&
\mathcal{D}^{(0)}_C[\Psi](x),
\end{array}
\end{displaymath}
where the harmonic function $H$ is given by
\begin{displaymath}
\begin{array}{rccccc}
\forall x \in \mathbb R^2 \setminus \partial \Omega, & \quad H(x)
&=& - \mathcal{S}^{(0)}_{\Omega}[g_{\mathrm{ele}}](x) &+&
\mathcal{D}^{(0)}_{\Omega}[u|_{\partial \Omega}](x),
\end{array}
\end{displaymath}
and $\Psi \in \mathcal{C}^{2, \eta}(\partial C)$ satisfies the
integral equation:
\begin{displaymath}
\begin{array}{rccccc}
\Psi &+& \beta \displaystyle \frac{\partial
\mathcal{D}^{(0)}_C[\Psi]}{\partial \nu} &=& - \beta \displaystyle
\frac{\partial H}{\partial \nu} \quad &\textrm{on} \, \partial C.
\end{array}
\end{displaymath}

We compute the first order approximation of
$\widetilde{c_{\mathrm{flr}}}$ using exactly the same method as in
Subsection \ref{subsectcell}. Doing so, we arrive with the help of
Corollary \ref{prop437} at
\begin{displaymath}
\widetilde{c_{\mathrm{flr}}} = c_{\mathrm{flr}} - \epsilon\,
\delta \,\Psi^{(1)}_1+ o(\epsilon),
\end{displaymath}
 where the function $\Psi^{(1)}_1$ is defined by
  \begin{displaymath} \begin{array}{l}
 \Psi^{(1)}_1 = - \beta (I + \beta \mathcal{L})^{-1}
 \bigg(( - \displaystyle \tau {h'} \frac{\partial H}{\partial T} + h \frac{\partial^2 H}{\partial \nu^2}
 + \frac{\partial}{\partial \nu} \mathcal{D}^{(0)}_{\Omega}[w]
 - h \displaystyle \frac{\partial}{\partial \nu}
 \mathcal{D}^{(0)}_{C}[\Psi]  \\
 \nm
 \ds \qquad + \frac{\partial}{\partial T} \mathcal{K}^{(0)}_C [h \frac{\partial \Psi}{\partial T}] - \frac{\partial}{\partial T}
 h \frac{\partial}{\partial T} \mathcal{K}^{(0)}_C[\Psi] \bigg),
 \end{array}
 \end{displaymath}
and $w$ is the solution to the problem:
 \begin{equation}\label{eq:wf}
\left\{
\begin{array}{ll}
\vspace{0.25 cm}\Delta w = 0 & \textrm{in}\, C \cup \Omega \setminus \overline{C},\\
\vspace{0.25 cm}\displaystyle\frac{\partial w}{\partial \nu}\bigg|_{+} -
\frac{\partial w}{\partial \nu}\bigg|_{-} = - \displaystyle
\frac{\partial}{\partial T} h \frac{\partial \Psi}{\partial T} & \textrm{on}\, \partial C,\\
\vspace{0.25 cm}w\mid_{+} - w\mid_{-} -\beta \displaystyle \frac{\partial w}{\partial \nu}
 \bigg|_{-}= - \beta \, \bigg( \tau {h} \frac{\partial u}{\partial \nu}
 + \frac{\partial}{\partial T} h \frac{\partial u}{\partial T} \bigg|_{-}\bigg)& \textrm{on} \,\partial C,\\
\displaystyle \frac{\partial w}{\partial \nu}\bigg|_{\partial
\Omega} = 0 & \textrm{on}\, \partial \Omega .
\end{array}
\right.
\end{equation}
We then obtain an expansion of (\ref{eq:im2}) as
$\epsilon\rightarrow 0$.
\begin{prop}\label{prop441}
Integral (\ref{eq:im2}) admits at first order in $\epsilon$ the
following expansion:
\begin{equation}
\begin{array}{l}
\vspace{0.3 cm} \displaystyle \int_{\partial C_{\epsilon}}
\tilde{\gamma} \widetilde{c_{\mathrm{flr}}}(x)
 \Phi_{\mathrm{exc}}^{\,f}(x)\Phi_{\mathrm{exc}}^{\,g}(x) \, ds(x) =
 \displaystyle \int_{\partial C} \tilde{\gamma} c_{\mathrm{flr}}(x) \Phi_{\mathrm{exc}}^{\,f}(x)\Phi_{\mathrm{exc}}^{\,g}(x) \,
 ds(x)\\
\qquad \qquad \qquad + \,\epsilon \displaystyle \int_{a}^{b}
\tilde{\gamma} \left( A(t)\,c_{\mathrm{flr}}(t)\,h(t) - \delta \,
B(t)\, \Psi^{(1)}_1(t)\right) \,dt + o(\epsilon),
\end{array}
\end{equation}
where the functions $A$ and $B$ are given by
\begin{equation}
\begin{array}{l}
\vspace{0.3 cm} A = \displaystyle
\frac{d\Phi_{\mathrm{exc}}^{\,f}(t)}{dt}
\Phi_{\mathrm{exc}}^{\,g}(t)
+ \Phi_{\mathrm{exc}}^{\,f}(t) \frac{d\Phi_{\mathrm{exc}}^{\,g}(t)}{dt} - \tau(t) \Phi_{\mathrm{exc}}^{\,f}(t) \Phi_{\mathrm{exc}}^{\,g}(t),\\
B = \Phi_{\mathrm{exc}}^{\,f}(t) \Phi_{\mathrm{exc}}^{\,g}(t).
\end{array}
\end{equation}
\end{prop}

Let $f_1,\ldots, f_L,$ be a finite number of linearly independent
functions in $L^2(\partial \Omega)$. We introduce the functional
$\mathcal{J}$ defined on $\mathcal{C}^2([a,b])$ by
\begin{displaymath}
\mathcal{J}(h) = \sum_{i,j=1}^{L} \left | \int_{\partial \Omega}
f_i(I^{\,f_j}_{\mathrm{emt},\epsilon} - I^{\,f_j}_{\mathrm{emt}})
\, ds - \epsilon \,\displaystyle \int_{a}^{b} \tilde{\gamma}
\left( A_{i,j}(t)\,c_{\mathrm{flr}}(t)\,h(t) - \delta \,
B_{i,j}(t)\, \Psi^{(1)}_1(t)\right) \,dt  \right |^2,
\end{displaymath}
where the functions $A_{i,j}$ and $B_{i,j}$ are given by
\begin{displaymath}
\begin{array}{l}
\vspace{0.3 cm} A _{i,j}= \displaystyle
\frac{d\Phi_{\mathrm{exc}}^{\,f_i}(t)}{dt}
\Phi_{\mathrm{exc}}^{\,f_j}(t) + \Phi_{\mathrm{exc}}^{\,f_i}(t)
\frac{d\Phi_{\mathrm{exc}}^{\,f_j}(t)}{dt} - \tau(t)
\Phi_{\mathrm{exc}}^{\,f_i}(t)
\Phi_{\mathrm{exc}}^{\,f_j}(t),\\
B_{i,j} = \Phi_{\mathrm{exc}}^{\,f_i}(t)
\Phi_{\mathrm{exc}}^{\,f_j}(t).
\end{array}
\end{displaymath}

We reconstruct the shape deformation $h$ by minimizing the
functional $\mathcal{J}$ over $h$. In order to maximize the
resolution of the reconstructed images, we choose
$f_1,\ldots,f_L,$ such that the functions $A_{i,j}$ and $B_{i,j}$
for $i,j \in [1,L]$ are highly oscillating. We will then be able
to obtain a resolved reconstruction of the boundary changes $h$.

We introduce the operator $\Lambda : L^2(\partial \Omega) \to
L^2(\partial C)$ defined by
\begin{displaymath}
\forall f \in L^2(\partial \Omega), \forall z \in \partial C,
\qquad \Lambda[f](z) = \Phi_{\mathrm{exc}}^{\,f}|_{\partial C}(z)
= \int_{\partial \Omega}G_z(y)f(y)\, ds(y).
\end{displaymath}

The adjoint operator $\Lambda^{\star} : L^2(\partial C) \to
L^2(\partial \Omega)$ is given by
\begin{displaymath}
\forall q \in L^2(\partial C), \forall y \in \partial \Omega,
\qquad \Lambda^{\star}[q](y) = p|_{\partial \Omega} (y) =
\int_{\partial C}\overline{G_z(y)}\, q(z)\, ds(z),
\end{displaymath}
where $p$ is the solution to the problem:
\begin{equation}
\left\{
\begin{array}{lll}
\vspace{0.3 cm} -\Delta p + k^2 p = 0 & \textrm{in} & \Omega ,\\
\vspace{0.3 cm} \displaystyle \frac{\partial p}{\partial \nu}\bigg|_{+}
- \,\frac{\partial p}{\partial \nu}\bigg|_{-}= - \,q \quad& \textrm{on} & \partial C , \\
\vspace{0.3 cm}p|_{+} - \,p|_{-}= 0 & \textrm{on} & \partial C \\
\displaystyle \ell \frac{\partial p}{\partial \nu} +  \,p= 0 & \textrm{on} & \partial \Omega .\\
\end{array}
\right.
\end{equation}

We therefore obtain the following expression for $\Lambda^{\star}
\Lambda$:
\begin{displaymath}
\forall f \in L^2(\partial \Omega), \forall y \in \partial \Omega,
\qquad \Lambda^{\star} \Lambda[f](y) =  \int_{\partial \Omega} dt
f(t) \int_{\partial C} \overline{G_z(y)}G_z(t) \, ds(z).
\end{displaymath}

Following \cite{elisa, garapon}, we choose $f_1,\ldots, f_L,$ to
be the first singular vectors of the operator $\Lambda$. The
number $L$, which fixes the resolving power of the approach, is
chosen to maximize the trade-off between resolution and stability.
To gain resolution one has to choose $L$ as large as possible. But
if it is too large then it follows from the fact that $f_i$ is
highly oscillating for large $i$ that the algorithm is unstable in
the case of noisy data \cite{garapon, solna}.

\section{Conclusion}
In this paper we have introduced and analyzed a mathematical model
for optical imaging of cell membrane potentials changes induced by
applied currents. We have presented a direct imaging algorithm in
the linearized case and provided explicit formulas for its
resolving power of the measurements in the presence of measurement
noise. We have suggested an iterative algorithm for complex
shapes. It would be interesting to consider the case of cluttered
cells. Another challenging problem is the tracking of membrane
changes in cell mechanisms such as cell division. This would  be
the subject of a forthcoming work.

\appendix
\section{Explicit calculation of $G_z$ in the case of a sphere} \label{appendixA}

We consider, in this appendix, that the dimension is three and
$\Omega$ is the unit sphere. We expand $G$, the solution to
(\ref{eq:greenex}),  in spherical harmonics $(Y^l_m)$:
$$ \forall z \in \Omega, \,\forall y (1, \theta, \phi) \in \partial
\Omega, \qquad G_z(y) = \sum_{l = 0}^{\infty}\, \sum_{m=-l}^{l}
g_{m}^{l,z}\, Y^l_m(\theta, \phi).$$

An addition theorem \cite[Formula (10-1-45/46)]{12} gives us the
expansion of $\Gamma$:
$$\forall z(r',\theta', \phi')\! \in\!\Omega, \forall y \in \partial \Omega, \quad \Gamma_z(y)
=ik \sum_{l = 0}^{\infty}\, \sum_{m=-l}^{l}  j_l(ikr')\,
h^{(1)}_l(ik)Y^l_m(\theta', \phi')\,Y^l_m(\theta, \phi),$$ where
$j_l$ and $\vspace{0.5mm}h^{(1)}_l$ are respectively the spherical
Bessel and Hankel functions of first kind of order $l$.

We then express the operators $\mathcal{S}_{\Omega}$ and
$\mathcal{K}_{\Omega}$ in terms of spherical harmonics \cite {2},
in the same way we wrote in the previous section their Fourier
coefficients:
\begin{displaymath}
\begin{array}{ccl}
 \vspace{1.2 mm}\forall y \in \partial \Omega,  \qquad &(-\displaystyle
  \frac{I}{2} + \mathcal{K}_{\Omega})[q](y) &=  - \displaystyle\sum_{l = 0}^{\infty}\,
   \sum_{m=-l}^{l} k^2 \,j_l^{'}(ik) \,h^{(1)}_l(ik) \,q_m^l \,Y^l_m(\theta, \phi),\\
\vspace{1.2 mm}\forall y \in \partial \Omega , \qquad
&\mathcal{S}_{\Omega}[q](y) &=  i \displaystyle \sum_{l =
0}^{\infty}\, \sum_{m=-l}^{l} k\, j_l(ik)\, h^{(1)}_l(ik)\, q_m^l
\,Y^l_m(\theta, \phi),
\end{array}
\end{displaymath}
for
\begin{displaymath}
\forall y(1,\theta, \phi) \in \partial \Omega, \quad q (y) =
\sum_{l = 0}^{\infty}\, \sum_{m=-l}^{l} q_m^l \,Y^l_m(\theta,
\phi).
\end{displaymath}

From (\ref{eqprop314}) we obtain
$$g_m^{l,z} = \displaystyle \frac{ ik\,j_l(ikr')\, h^{(1)}_l(ik)Y^l_m(\theta', \phi')}{-k^2 \,j_l^{'}(k)
 \,h^{(1)}_l(ik) + \frac{1}{\ell}\, ik\, j_l(ik)\, h^{(1)}_l(ik)} = \frac{ j_l(ikr')}{ik \,j_l^{'}(ik)
 + \frac{1}{\ell} \,j_l(ik)\,}Y^l_m(\theta', \phi'),$$
or else, for all $z = (r',\theta', \phi') \in \Omega$ and $y =(1,
\theta, \phi) \in \partial \Omega$,
$$G_z(y) = \sum_{l = 0}^{\infty}\, \sum_{m=-l}^{l} \displaystyle
\frac{ j_l(ikr')}{ik \,j_l^{'}(ik) + \frac{1}{\ell}
\,j_l(ik)\,}Y^l_m(\theta', \phi') Y^l_m(\theta, \phi).$$

Note that we find a very similar formula as the one in $2$D. The
Bessel function of first kind is replaced by the spherical
function of first kind, and our operator is decomposed in the
spherical harmonics basis instead of the Fourier basis.

\end{document}